\documentclass[10pt,journal,cspaper,compsoc,twoside]{IEEEtran}
\usepackage{graphicx}
\usepackage{balance}  % for  \balance command ON LAST PAGE  (only there!)
\usepackage{algorithmic}
\usepackage{algorithm}
\usepackage{url}
\usepackage{xspace}
\usepackage{ifthen}
\usepackage{array}
\usepackage{color}
\usepackage{cite}
\usepackage{amsmath}
\usepackage{amsfonts}

\newboolean{isDoubleBlind}
\setboolean{isDoubleBlind}{false}
\newcommand{\singleDoubleBlind}[2]{\ifthenelse{\boolean{isDoubleBlind}}{#2}{#1}}

\newcommand{\cf}[0]{\emph{cf.}\xspace}
\newcommand{\eg}[0]{\emph{e.g.},\xspace}
\newcommand{\ie}[0]{\emph{i.e.},\xspace}
\newcommand{\st}[0]{\ensuremath{\mbox{s.t.}}}
\newcommand{\etal}[0]{\emph{et al.}\xspace}
\newcommand{\term}[1]{\emph{#1}}
\newcommand{\nth}[2]{\ensuremath{{#1}^{\mbox{\scriptsize #2}}}}
\def\ignore#1{}
\newcommand{\AlgMatchMM}[0]{\textsc{ManyMany}}
\newtheorem{theorem}{Theorem}[section]
\newtheorem{definition}[theorem]{Definition}

\newtheorem{example}[theorem]{Example}
\ifCLASSOPTIONcompsoc
\else
\fi
\ifCLASSINFOpdf
\else
\fi
\begin{document}
\title{Principled Graph Matching Algorithms for \\ Integrating Multiple Data Sources\thanks{Research performed by authors while at Microsoft Research.}}

%Duo Zhang \hspace{3em} Benjamin I. P. Rubinstein \hspace{3em} Jim Gemmell \\
%       \affaddr{\hspace{3.5em} Twitter, Inc. \hspace{7.0em} IBM Research \hspace{10.0em} Tr\={o}v\hspace{5em}} \\
%       \affaddr{\hspace{3em} dzhang@twitter.com \hspace{6.0em} ben@bipr.net \hspace{5.8em} jim.gemmell@gmail.com\hspace{2.1em}}

\author{Duo~Zhang,
        Benjamin~I.~P.~Rubinstein,
        and~Jim~Gemmell% <-this % stops a space
\IEEEcompsocitemizethanks{\IEEEcompsocthanksitem D. Zhang is with Twitter, Inc., USA.
\IEEEcompsocthanksitem B. Rubinstein is with the University of Melbourne, Australia.
\IEEEcompsocthanksitem J. Gemmell is with Tr\={o}v, USA.\protect\\
E-mail: dzhang@twitter.com, ben@bipr.net, jim.gemmell@gmail.com}% <-this % stops a space
\thanks{}}


\markboth{IEEE Transactions on Knowledge and Data Engineering, Vol. 26, No. X, X 2014}%
{Zhang et al.: Principled Graph Matching Algorithms for Integrating Multiple Data Sources}

%\IEEEpubid{\makebox[\columnwidth]{\hfill 0000--0000/00/\$00.00~\copyright~2007 IEEE}%

% PRIOR to the title within the \IEEEcompsoctitleabstractindextext IEEEtran
\IEEEcompsoctitleabstractindextext{%
\begin{abstract}
This paper explores combinatorial optimization for problems of max-weight graph matching
on multi-partite graphs, which arise in integrating multiple data sources.
Entity resolution---the data integration problem of performing noisy joins on structured
data---typically proceeds by first hashing each record into zero or more blocks, scoring pairs
of records that are co-blocked for similarity, and then matching pairs of sufficient similarity.
In the most common case of matching two sources, it is often desirable for the final
matching to be one-to-one (a record may be matched with at most one other); members of the
database and statistical record linkage communities accomplish such matchings in the final stage
by weighted bipartite graph matching on similarity scores. Such matchings are intuitively appealing:
they leverage a natural global property of many real-world entity stores---that of being nearly deduped---and
are known to provide significant improvements to precision and recall. Unfortunately unlike the
bipartite case, exact max-weight matching on multi-partite graphs is known to be NP-hard. Our two-fold
algorithmic contributions approximate multi-partite max-weight matching:
our first algorithm borrows optimization techniques common to Bayesian probabilistic inference;
our second is a greedy approximation algorithm. In addition to a theoretical guarantee on the latter,
we present comparisons on a real-world ER problem from Bing significantly larger than
typically found in the literature, publication data, and on a series of synthetic problems.
Our results quantify significant improvements due to exploiting multiple
sources, which are made possible by global one-to-one constraints linking otherwise independent
matching sub-problems. We also discover that our algorithms are complementary: one being much more robust
under noise, and the other being simple to implement and very fast to run.
\end{abstract}

% not want either math or citations to appear in the abstract.

% Note that keywords are not normally used for peer review papers.
\begin{keywords}
Data integration, weighted graph matching, message-passing algorithms
\end{keywords}}

\maketitle

\section{Introduction}\label{sec:intro}

It has long been recognized---and explicitly discussed recently~\cite{1to1}---that
many real-world entity stores are naturally free of duplicates.
Were they to have replicate entities, crowd-sourced sites such as Wikipedia would have edits applied to one copy and not others.
Sites that rely on ratings such as Netflix and Yelp would suffer diluted recommendation value by having ratings split over multiple instantiations of the same product or
business page.
And customers of online retailers such as Amazon would miss out on
lower prices, options on new/used condition, or shipping arrangements
offered by sellers surfaced on duplicate pages.
%We have previously estimated the level of duplication in a large online movies database used in experiments here, to be less than 0.1\%~\cite{1to1}.
Many publishers have natural incentives that drive them to deduplicate, or maintain uniqueness,
in their databases.

Dating back to the late 80s in the statistical record linkage community and more recently
in database research~\cite{WebLists,Jaro89,ConstraintBased2005,Challenges,1to1,CMU_record}, a
number of entity resolution (ER) systems have successfully employed one-to-one graph matching
for leveraging this natural lack of duplicates. Initially the benefit
of this approach was taken for granted, but in preliminary recent
work~\cite{1to1} significant improvements to precision and recall due to this approach have
been quantified.
The reasons are intuitively clear: data noise or deficiencies of the scoring function can lead to poor scores, which can negatively affect ER accuracy; however graph matching corresponds to imposing a global, one-to-one constraint which effectively smoothes local noise.
This kind of bipartite one-to-one matching for combining a pair of data sources is both well-known to many in the community, and widely
applicable through data integration problems as it can be used to augment numerous existing systems---\eg~\cite{InfoSys-MultClassSys,JCoopInfoSys,QDBMUD08-STEM,Benjelloun09,Kopcke2010,Matching_product,JASA69,Winkler94,PAKDD08-FEBRL,Winkler06,ActiveLearn-KDD02,KDD03-MARLIN,Culotta2005,DedupSurvey07}---as
a generic stage following data preparation, blocking, and scoring.

\begin{figure*}[t]
\begin{center}
\begin{minipage}[t]{0.31\textwidth}
\centering
\includegraphics[width=1\linewidth]{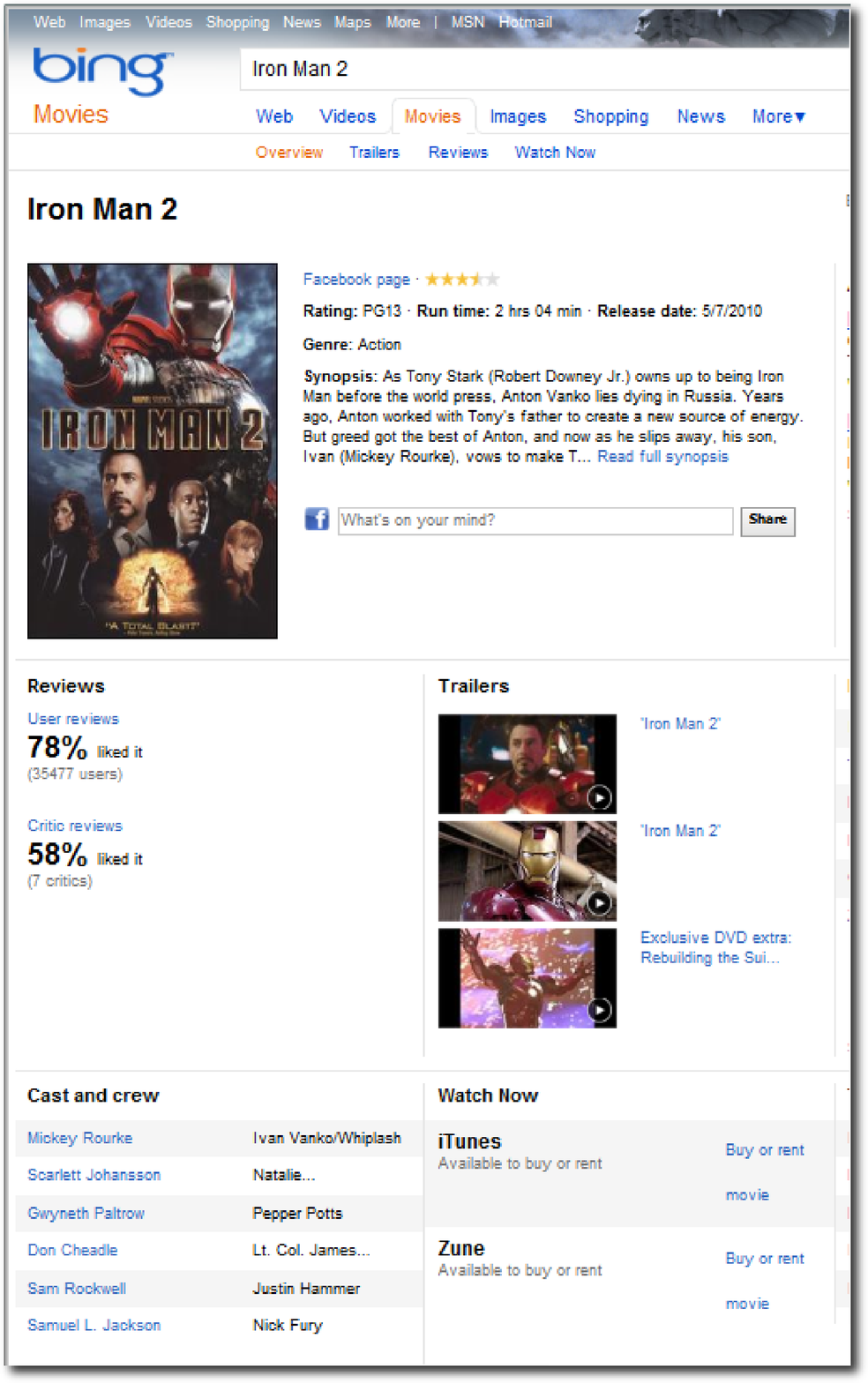}
%\caption{Bing.com.}
%\label{fig:bing}
\end{minipage}\hfill
\begin{minipage}[t]{0.31\textwidth}
\centering
\includegraphics[width=1\linewidth]{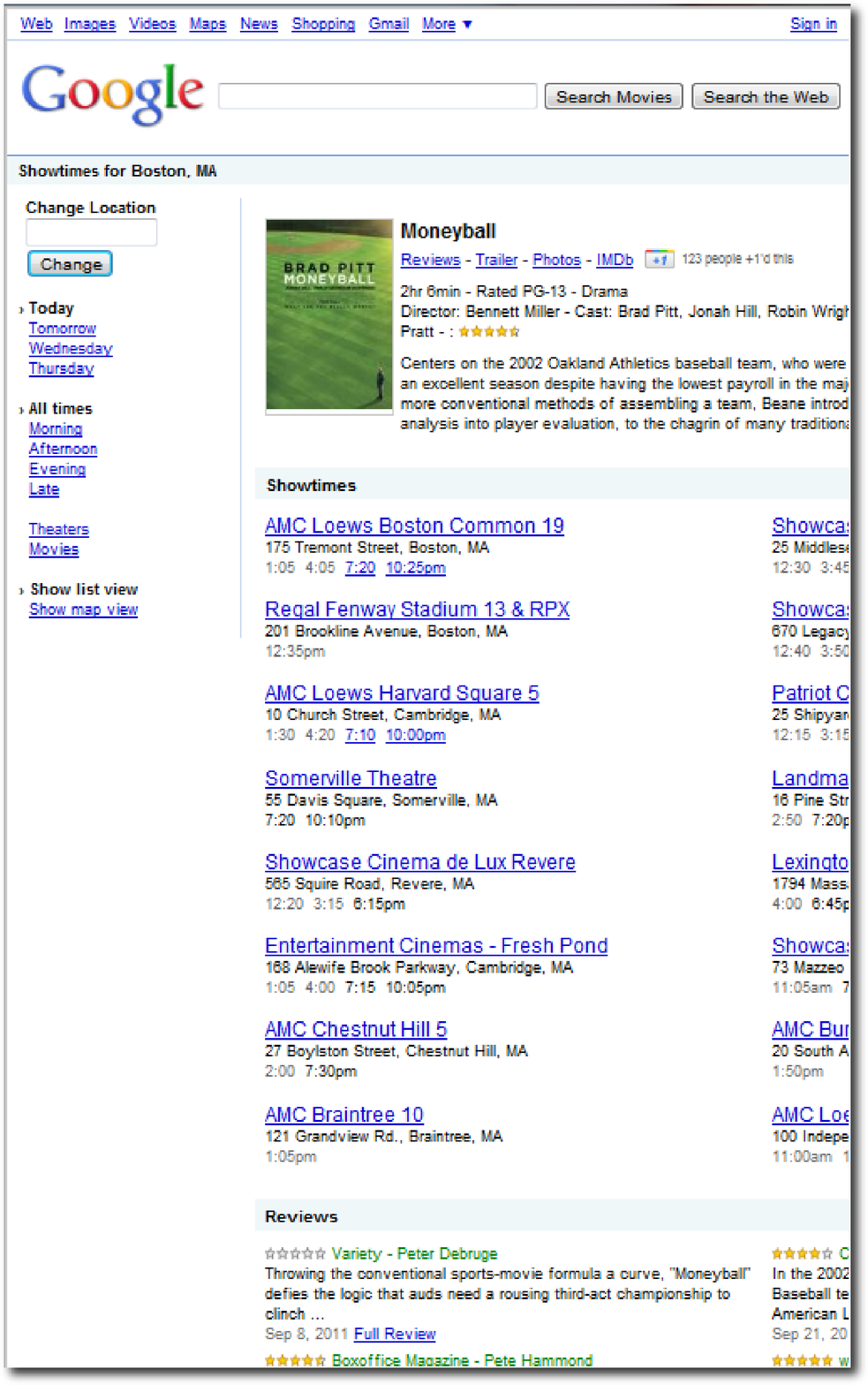}
%\caption{Google.com.}
%\label{fig:google}
\end{minipage}\hfill
\begin{minipage}[t]{0.31\textwidth}
\centering
\includegraphics[width=1\linewidth]{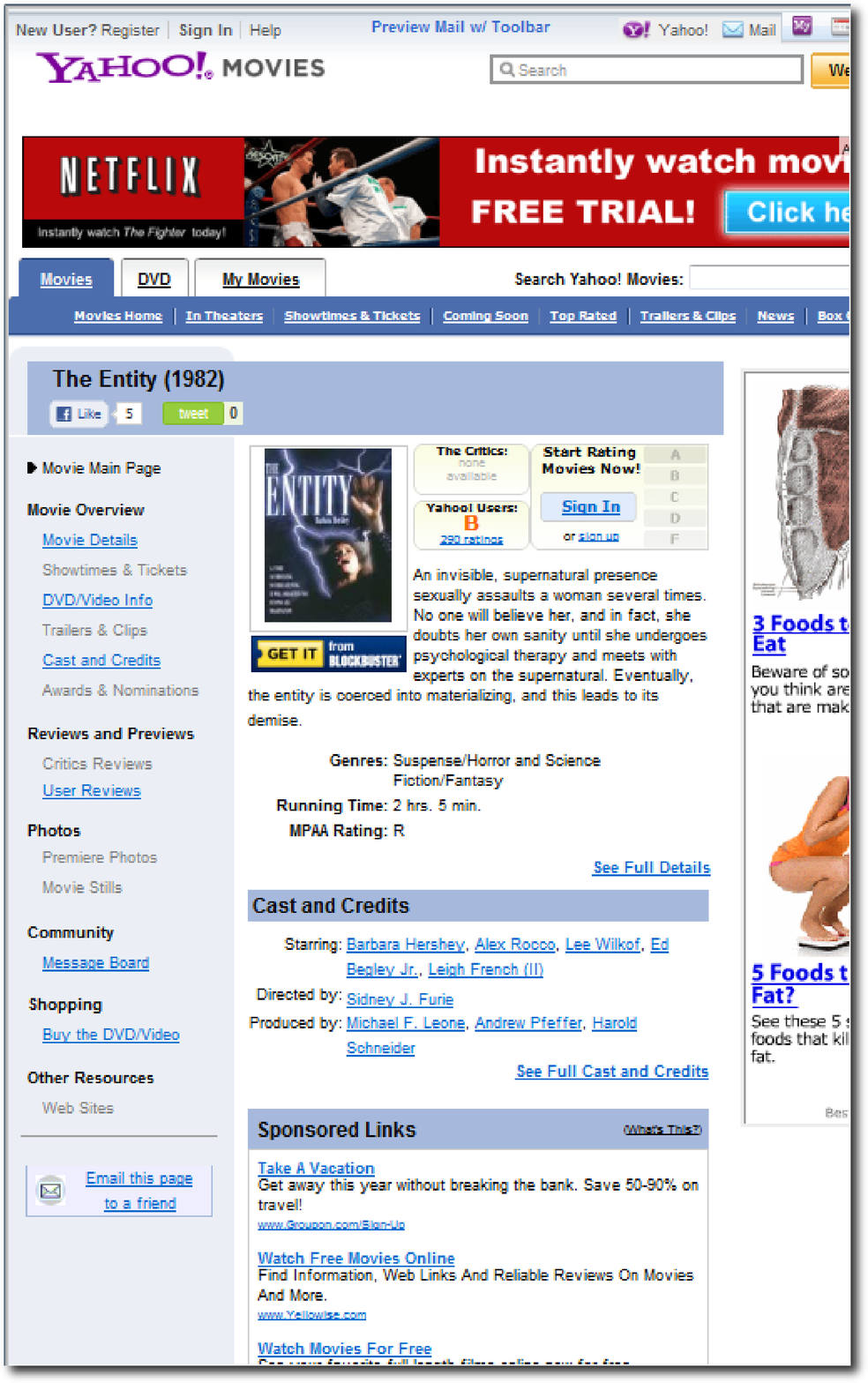}
%\caption{Yahoo.com.}
%\label{fig:yahoo}
\end{minipage}
%\singleDoubleBlind{
%\caption{Multi-source entity resolution (left) is used in Bing's movies vertical for implementing ``buy'', ``rent'', ``watch'', ``buy ticket'' entity actions and surfacing meta-data, showtimes and reviews, and (middle) could be used for similar purposes in Google's \& (right) Yahoo!'s movies verticals.}
%}{
\caption{Multi-source ER is central to (left) Bing's (middle) Google's and (right) Yahoo!'s movies verticals for implementing entity actions such as ``buy'', ``rent'', ``watch'', ``buy ticket'' and surfacing meta-data, showtimes and reviews. In each case data on the same movie is integrated from multiple noisy sources.}
%}
\end{center}
\label{fig:bing}
\end{figure*}

Another common principle to data integration is of improved user utility by fusing multiple sources of data.
But while many papers have circled around the problem, very little is known about extending
one-to-one graph matching to practical multi-partite matching for ER. In the theory community
exact max-weight matching is known to be NP-hard~\cite{Crama1992273}, and in the statistics community expectation
maximization has recently been applied to approximate the solution successfully for ER, but inefficiently with exponential computational requirements~\cite{CMU_record}.

In this paper we propose principled approaches for approximating multi-partite weighted graph matching. Our first approach is based on message-passing
algorithms typically used for inference on probabilistic graphical models but used here for
combinatorial optimization. Through a series of non-trivial approximations we derive an approach
more efficient than the leading statistical record linkage work of Sadinle~\etal~\cite{CMU_record}. Our second approach extends the
well-known greedy approximation to bipartite max-weight matching to the multi-partite case. While
less sophisticated than our message-passing algorithm, the greedy approach enjoys an
easily-implementable design and a worst-case 2-approximation competitive ratio (\cf Theorem~\ref{thm:np-cr}).

The ability to leverage one-to-one constraints when performing multi-source data integration is
of great economic value: systems studied here have made significant impact in production use for example within several Bing verticals
and the Xbox TV service driving critical customer-facing features (\cf \eg Figure~\ref{fig:bing}).
We demonstrate on data taken from these real-world services that our approaches enjoy
significantly improved precision/recall over the state-of-the-art unconstrained 
approach and that the addition of sources yields further improvements due to global constraints.
This experimental study is of atypical value owing to its unusually large scale:
compared to the largest of the four datasets used in the recent well-cited evaluation study
of K\"opcke~\etal~\cite{KopckeTR10}, our problem is three orders of magnitude larger.\footnote{While their largest problem contains $1.7\times 10^8$ pairs, our's measures in at $1.6\times 10^{11}$ between just two of our sources. For data sources measuring in the low thousands, as in their other benchmark problems and as is typical in many papers, purely crowd-sourcing ER systems such as~\cite{HumanJoins} could be used for mere tens of dollars.}
We conduct a second experimental comparison on a smaller publication dataset to
demonstrate generality.
Finally we explore the robustness of our approaches to varying degrees of edge weight noise via synthetic
data.

In summary, our main contributions are:
\begin{enumerate}

\item A principled factor-graph message-passing algorithm for generic, globally-constrained multi-source data integration;

\item An efficient greedy approach that comes with a sharp, worst-case guarantee (\cf Theorem~\ref{thm:np-cr});

\item A counter example to the common misconception that sequential bipartite matching is a sufficient approach to joint multipartite matching (\cf Example~\ref{eg:sequential});

\item Experimental comparisons on a very large real-world dataset demonstrating that generic one-to-one matching leverages naturally-deduplicated data to significantly improve precision/recall;

\item Validation that our new approaches can appropriately leverage information from multiple sources to improve ER precision and recall---further supported by a second smaller real-world experiment on publication data; and

\item A synthetic-data comparison under which our message-passing approach
enjoys superior robustness to noise, over the faster greedy approach.

\end{enumerate}

\section{Related Work}\label{sec:related}

Numerous past work has studied entity resolution~\cite{InfoSys-MultClassSys,JCoopInfoSys,QDBMUD08-STEM,Benjelloun09,Kopcke2010,Matching_product}, statistical record linkage~\cite{JASA69,Winkler94,PAKDD08-FEBRL,Winkler06} and deduplication~\cite{ActiveLearn-KDD02,KDD03-MARLIN,Culotta2005,DedupSurvey07}.
There has been little previous work investigating one-to-one bipartite resolutions and almost no results on constrained multi-source resolution.

Some works have looked at constraints in general~\cite{Chaudhuri2007,ConstraintBased2005,Dedupalog} but they do not focus on the one-to-one constraint. Su \etal\ do rely on the rarity of duplicates in a given web source~\cite{NegativeExamples} but use it only to generate negative examples. Guo \etal\ studied a record linkage approach based on uniqueness constraints on entity attribute values. In many domains, however, this constraint does not always hold. For example, in the movie domain, each movie entity could have multiple actor names as its attribute. Jaro produced early work~\cite{Jaro89} on linking census data records using a Linear Program formulation that enforces a global one-to-one constraint, however no justification or evaluation of the formulation is offered, and the method is only defined for resolution over two sources.  Similarly for more recent work in databases, such as in conflating web tables without duplicates~\cite{WebLists} and the notion of exclusivity constraints~\cite{Challenges}. We are the first to test and quantify the use of such a global constraint on entities and apply it in a systematic way building on earlier bipartite work~\cite{1to1}.

A more recent study~\cite{CMU_record} examines the record linkage problem in the multiple sources setting. The authors use a probabilistic framework to estimate the probability of whether a pair of entities is a match. However, the complexity of their algorithm is $O(B_mn^m)$, where $m$ is the number of sources, $n$ is the number of instances, and $B_m$ is the \nth{m}{th} Bell number which is exponential in $m$. Such a prohibitive computational complexity prevents the principled approach from being practical. In our message-passing algorithm, the optimization approach has far better complexity, which works well on real entity resolution problems with millions of entities; at the same time our approach is also developed in a principled fashion. %Other past work on multi-source entity resolution~\cite{essay} studies the problem in a setting of pair-wise data integration that is associative; that paper's main focus is to resolve matching conflicts, which differs from the problem studied here.
Our message-passing algorithm is the first principled, tractable approach to constrained multi-source ER.

In the bipartite setting, maximum-weight matching is now well understood with exact algorithms able to solve the problem in polynomial time~\cite{PapaBook}. When it comes to multi-partite maximum-weight matching, the problem becomes extremely hard. In particular Crama and Spieksma~\cite{Crama1992273} proved that a special case of the tripartite matching problem is NP-hard, implying that the general multi-partite max-weight matching problem is itself NP-hard. The algorithms presented in this paper are both approximations: one greedy approach that is fast to run, and one principled approach that empirically yields higher total weights and improved precision and recall in the presence of noise. Our greedy approach is endowed with a competitive ratio generalizing the well-known guarantee for the bipartite case.

%While our specific problem has not been studied theoretically, a related problem of max weight matching on weighted multi-partite hypergraphs has, where hyperedges are endowed with weights and the goal is to select a subset of hyperedges of maximum weight while obeying the one-to-one constraint on nodes. It has been proven that even approximating solutions to the 3-partite problem is worst-case NP-complete~\cite{Hazan03}.

Previously Bayesian machine learning methods have been applied to entity resolution~\cite{Winkler02,Ravikumar04,Culotta2005,bhattacharya:sdm06}, however our goal here is not to perform inference with these methods but rather optimization. Along these lines in recent years, message-passing algorithms have been used for the maximum-weight matching problem in bipartite graphs~\cite{Bipartite}. There the authors proved that the max-product algorithm converges to the desirable optimum point in finding the maximum-weight matching for a bipartite graph, even in the presence of loops. In our study, we have designed a message-passing algorithm targeting the maximum-weight matching problem in multi-partite graphs. In a specific case of our problem, \ie maximum-weight matching in a bipartite graph, our algorithm works as effectively as exact methods. Another recent work~\cite{General_graph} studied the weighted-matching problem in general graphs, whose problem definition differs from our own. It shows that max-product converges to the correct answer if the linear programming relaxation of the weighted-matching problem is tight. Compared to this work we pay special attention to the application to multi-source entity resolution, tune the general message-passing algorithms specifically for our model, and perform large-scale data integration experiments on real data.

%\begin{figure}[t]
%\centering
%\begin{tabular}{>{\centering}p{2.6cm}>{\centering}p{2.6cm}>{\centering}p{2.6cm}}
%  \includegraphics[width=2.6cm]{figures/example_graph.eps} & \includegraphics[width=2.6cm]{figures/sequential_matching.eps} & \includegraphics[width=2.6cm]{figures/true_matching.eps} \tabularnewline[0.2em]
%  (a) & (b) & (c) \tabularnewline
%\end{tabular}
%\caption{(a) Example ER problem on three sources of three entities each (nodes) and pairwise scores of 0.2 for correct pairs (red) and 0.32 for noisy non-matching pairs (blue); (b) incorrect matching producing by thresholding and sequential bipartite matching; (c) the true matching achieves the max total weight of 1.8.}
%\label{fig:examples}
%\end{figure}

\section{The MPEM Problem}\label{sec:problem}

We begin by formalizing the generic data integration problem on multiple sources. Let $D_1$, $D_2$, $\dots$, $D_m$ be $m$ \term{databases}, each of which contains representations of a finite number of \term{entities} along with a special null entity $\phi$. The database sizes need not be equal.

\begin{definition}
Given $D_1$, $D_2$, $\dots$, $D_m$, the \term{Multi-Partite Entity Resolution} problem is to identify an unknown target relation or \term{resolution} $R \subseteq D_1 \times D_2 \times \dots \times D_m$, given some information about $R$, such as pairwise scores between entities, examples of matching and non-matching pairs of entities, etc.
\end{definition}

%In the definition, $D_1 \times D_2 \times \dots \times D_m$ means all possible mappings among entities in the $m$ databases.
For example, for databases $D_1=\{e_1,\phi\}$ and $D_2=\{e'_1, e'_2, \phi\}$, the possible mappings are $D_1\times D_2=\{e_1 \leftrightarrow e'_1, e_1 \leftrightarrow e'_2, e_1 \leftrightarrow \phi, \phi \leftrightarrow e'_1, \phi \leftrightarrow e'_2, \phi \leftrightarrow \phi\}$. Here, $\phi$ appearing in the \nth{i}{th} component of a mapping $r$ means that $r$ does not involve any entity from source $i$. A resolution $R$, which represents the true matchings among entities, is some specific subset of all possible mappings among entities in the $m$ databases.

A \term{global one-to-one constraint} on multi-partite entity resolution asserts the target $R$ is \term{pairwise one-to-one}: for each non-null entity $x_i \in D_i$, for all sources $j \neq i$ there exists at most one entity $x_j \in D_j$ such that together $x_i, x_j$ are involved in one tuple in $R$ .

\begin{definition}
A \term{Multi-Partite Entity Matching (MPEM)} problem is a multi-partite entity resolution problem with the global one-to-one constraint.
\end{definition}

We previously showed experimentally~\cite{1to1} that leveraging the one-to-one constraint when performing entity resolution across two sources yields
significantly improved precision and recall---a fact long well-known anecdotally in the databases and statistical record linkage communities.
For example, if we know \emph{The Lord of the Rings I} in IMDB matches with \emph{The Lord of the Rings I} in Netflix, then the one-to-one property precludes the possibility of this IMDB movie resolving with \emph{The Lord of the Rings II} in Netflix.

The present paper focuses on the MPEM problem and methods for exploiting the global constraint when
resolving across multiple sources. In particular, we are interested in global methods which perform entity resolution across all sources simultaneously. In the experimental section, we will show that matching multiple sources together achieves superior resolution as compared to matching two sources individually (equivalently pairwise matching of multiple sources).

\subsection{The Duplicate-Free Assumption}
\label{sec:deduped}

As discussed in Section~\ref{sec:intro}, many real-world data sources are naturally de-duplicated due to various socio-economic incentives. There we argued, by citing significant online data sources such as Wikipedia, Aamzon, Netflix, Yelp, that
\begin{itemize}
\item Crowd-sourced site duplicates will diverge with edits applied to one copy and not the other;
\item Sites relying on ratings suffer from duplicates as recommendations are made on too little data; and
\item Attributes will be split by duplicates, for example alternate prices of shipping options of products.
\end{itemize}

Indeed in our past work, we quantified that the data used in this paper, coming (raw and untouched) from major online movie
sources, are largely duplicate free, with estimated levels of 0.1\% of duplicates. Our publication dataset is similar in terms of duplicates.

Finally, consider matching and merging multiple sources. Rather than taking a union of all sources then
performing matching by deduplication on the union, by instead (1) deduplicating each source then (2) matching, 
we may focus on heterogeneous data characteristics in the sources individually. Recognizing the importance
of data variability aids matching accuracy when scaling to many sources~\cite{transfer}.

\subsection{Why MPEM Requires Global Methods}

One may naively attempt to use one of two natural approaches for the MPEM problem: (1) threshold the overall scores to produce many-to-many resolutions; or (2) 1-to-1 resolve the sources in a sequential manner by iteratively performing two-source entity resolution. The disadvantages of the first have been explored previously~\cite{1to1}. The order of sequential bipartite matching may adversely affect overall results: if we first resolve two poor quality sources, then the poor resolution may propagate to degrade subsequent resolutions. We could undo earlier matchings in possibly cascading fashion, but this would essentially reduce to global multipartite. 

The next example demonstrates the inferiority of sequential matching.

\begin{example}\label{eg:sequential}
Take weighted tri-partite graph with nodes $a_1, b_1, c_1$ (source $1$), $a_2, b_2, c_2$ (source $2$), $a_3, b_3, c_3$ (source $3$). The true matching is $\{\{a_1, a_2, a_3\}, \{b_1, b_2, b_3\}, \{c_1, c_2, c_3\}\}$.
Each pair of sources is completely connected, with weights:
%\begin{itemize}
0.6 for all pairs between $1$ and $2$ except 0.5 between $\{a_1, a_2\}$ and 1 between $\{c_1, c_2\}$; 
1 for records truly matching between $1$ or $2$, and $3$;
0.1 for records not truly matching between $1$ or $2$, and $3$.
%\end{itemize}
When bipartite matching between $1$ and $2$ then $3$ (preserving transitivity
in the final matching), sequential achieves weight $6.4$ while global matching achieves $8.1$.
\end{example}

\textit{
Sequential matching can produce very inferior overall matchings
due to being locked into poor subsequent decisions \ie local optima. This is compounded by poor quality data. In contrast global matching looks ahead.} As a result we focus on global approaches.

%As an example, consider the weighted multi-partite graph shown in Figure~\ref{fig:examples}(a), where we have three sources $A$, $B$, and $C$. Suppose the true matchings are $(a_i, b_i, c_i)$, where $i=1,2,3$, and all the edges between matching entities are scored as 0.2 (shown in red). Due to noise in the data, there are also some high scores of 0.32 across the sources, shown as blue edges. If we use the straw man threshold approach, we may either keep all the edges in the graph, which violates the one-to-one constraint, or keep only the incorrect noisy edges, as shown in Figure~\ref{fig:examples}(b). Similarly, if we use the sequential matching approach and first resolve sources $A$ and $B$, we will also achieve the result of Figure~\ref{fig:examples}(b). Since the total number of bi-partite matching sequences is the factorial of the number of sources, the complexity of finding an optimal sequential order (in this case matching $A,C,B$) is too high in general. And in general such an optimal choice may not even yield favorable results.

%Together this reasoning implies that multi-source ER should ideally be performed across all sources simultaneously. In the next section we present two approaches to the multi-partite matching problem for multi-source data integration, both of which try to maximize the total weight of the global resolution.
%Turning to our three-source example Figure~\ref{fig:examples}(c) shows that the maximum weight strategy will indeed find the true matching. 

\section{Algorithms}\label{sec:algorithms}

We now develop two strategies for solving MPEM problems. The first is to approximately maximize total weight of the matching by co-opting optimization algorithms popular in Bayesian inference. The second approach is to iteratively select the best matching for each entity locally while not violating the one-to-one constraint, yielding a simple greedy approach.

\subsection{Message-Passing Algorithm}

As discussed in Section~\ref{sec:related}, exact multi-partite maximum-weight matching is NP-hard. Therefore we develop an approximate yet principled maximum-weight strategy for MPEM problems.  Without loss of generality, we use tri-partite matching from time-to-time for illustrative purposes.

We make use of the max-sum algorithm that drives \emph{maximum a posteriori}
probabilistic inference on Bayesian graphical models~\cite{Koller+Friedman:09}. Bipartite factor graphs
model random variables and a joint likelihood on them, via variable nodes
and factor nodes respectively. Factor nodes correspond to a factorization
of the unnormalized joint likelihood, and are connected to variable nodes iff
the corresponding function depends on the corresponding variable. \emph{Max-sum}
maximizes the sum of these functions (\ie MAP inference in log space), and
decomposes via dynamic programming to an iterative message-passing
algorithm on the graph.

In this paper we formulate the
MPEM problem as a maximization of a sum of functions over variables, naturally
reducing to max-sum over a graph.

\subsubsection{Factor Graph Model for the MPEM Problem}
The MPEM problem can be represented by a graphical model, shown in part in Figure~\ref{fig:graphical_model} for tri-partite graph matching. In this model, $C_{ijk}$ represents a variable node, where $i,\, j,\, k$ are three entities from three different sources. We use $x_{ijk}$ to denote the variable associated with graph node $C_{ijk}$. This boolean-valued variable selects whether the entities are matched or not. If $x_{ijk}$ is equal to 1, it means all three entities are matched. Otherwise, it means at least two of these three entities are not matched. Since we will not always find a matching entity from every source, we use $\phi$ to indicate the case of no matching entity in a particular source. For example, if variable $x_{i,j,\phi}$ is equal to 1, then $i$ and $j$ are matched without a matching entity from the third source.

\begin{figure}[h]
\begin{centering}
\includegraphics[scale=0.3]{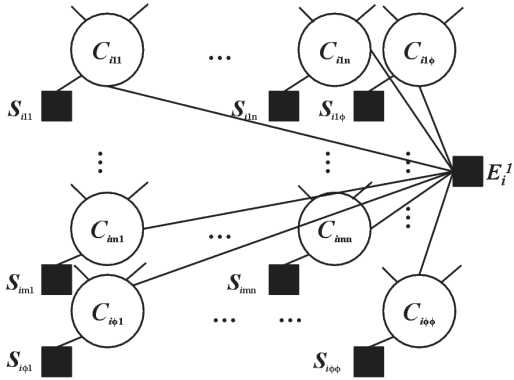}
\par\end{centering}
\caption{Graphical model for the MPEM problem.}
\label{fig:graphical_model}
\end{figure}

Figure~\ref{fig:node_view} shows all the factor nodes connected with a variable node $C_{ijk}$.
There are two types of factor nodes: (1) a similarity factor node $S_{ijk}$ which represents the total weight from matching $i,\, j,\, k$; (2) three constraint factor nodes $E_{i}^{1}$, $E_{j}^{2}$, and $E_{k}^{3}$ which represent one-to-one constraints on entity $i$ from source 1, entity $j$ from source 2, and entity $k$ from source 3, respectively. The similarity factor node only connects to a single variable node, but the constraint factor nodes connect with multiple variable nodes.

\begin{figure}[h]
\begin{centering}
\includegraphics[scale=0.3]{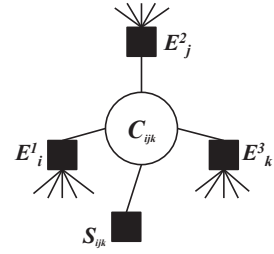}
\par\end{centering}
\caption{Variable node's neighboring factor nodes.}
\label{fig:node_view}
\end{figure}

The concrete definitions of the factor nodes are shown in Figure~\ref{fig:factor_nodes}. The definition of the similarity function is easy: if a variable node is selected, the function's value is the total weight between all pairs of entities within the variable node. Otherwise, the similarity function is 0. Each constraint $E$ is defined on a ``plane'' of variables which are related to one specific entity from a specific source. For example, as shown in Figure~\ref{fig:graphical_model}, the factor $E^1_i$ is defined on all the variables related to entity $i$ in source 1. It evaluates to 0 \emph{iff} the sum of the values of all the variables in the plane is equal to 1, which means exactly one variable can be selected as a matching from the plane. Otherwise, the function evaluates to negative infinity, penalizing the violation of the global one-to-one constraint. Thus these factor nodes serve as hard constraints enforcing global one-to-one matching. Note that there is no factor $E^1_\phi$ defined on the variables related to entity $\phi$ from a source, because we allow the entity $\phi$ to match with multiple different entities from other sources.

%\begin{figure*}[!ht]
%\[
%S_{ijk}(x_{ijk})=\begin{cases}
%s_{ij}+s_{jk}+s_{ik}, & x_{ijk}=1\\
%0, & otherwise
%\end{cases}
%\]
%\[
%E_{i}^{1}(x_{i11},...,x_{i1\phi},...,x_{ijk},...,x_{i\phi\phi})=\begin{cases}
%0 & x_{i11}+...+x_{i1\phi}+...+x_{ijk}+...+x_{i\phi\phi}=1\\
%-\infty & otherwise
%\end{cases}
%\]
%
%\[
%E_{j}^{2}(x_{1j1},...,x_{1j\phi},...,x_{ijk},...,x_{\phi j\phi})=\begin{cases}
%0 & x_{1j1}+...+x_{1j\phi}+...+x_{ijk}+...+x_{\phi j\phi}=1\\
%-\infty & otherwise
%\end{cases}
%\]
%\[
%E_{k}^{3}(x_{11k},...,x_{1\phi k},...,x_{ijk},...,x_{\phi\phi k})=\begin{cases}
%0 & x_{11k}+...+x_{1\phi k}+...+x_{ijk}+...+x_{\phi\phi k}=1\\
%-\infty & otherwise
%\end{cases}
%\]
%\caption{Definitions of factor nodes' potential functions displayed in Figure~\ref{fig:node_view}.}
%\label{fig:factor_nodes}
%\end{figure*}

\begin{figure}[h]
\begin{centering}
\includegraphics[width=\columnwidth]{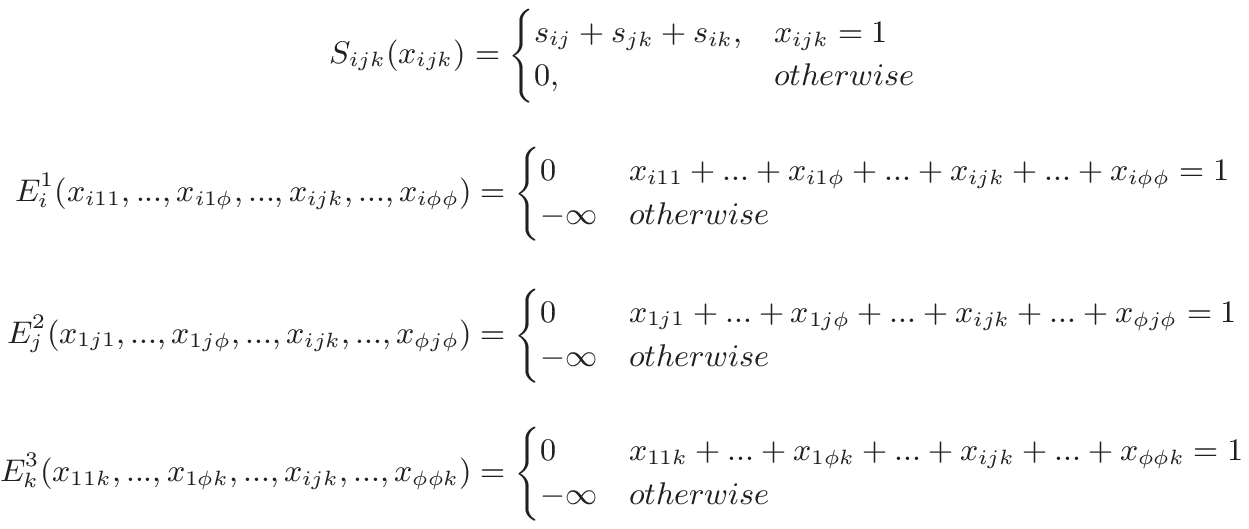}
\par\end{centering}
\caption{Definitions of factor nodes' potential functions displayed in Figure~\ref{fig:node_view}.}
\label{fig:factor_nodes}
\end{figure}

Assuming all functions defined in Figure~\ref{fig:factor_nodes} are in $\log$ space, then their sum corresponds to the objective function
\begin{equation}
f=\sum_{i,j,k}S_{ijk}(x_{ijk})+\sum_{s=1}^{3}\sum_{t=1}^{N_{s}}E_{t}^{s}\ ,
\end{equation}
where $N_s$ is the number of entities in source $s$. Maximizing this sum is equivalent to one-to-one max weight matching, since the constraint factors are \emph{barrier} functions. The optimizing assignment is the configuration of all variable nodes, either 0 or 1, which maximizes the objective function as desired.

%\begin{figure*}[ht]
%\[
%S_{ijk}(x_{ijk})=\begin{cases}
%s_{ij}+s_{jk}+s_{ik}, & x_{ijk}=1\\
%0, & otherwise
%\end{cases}
%\]
%\[
%E_{i}^{1}(x_{i11},...,x_{i1\phi},...,x_{ijk},...,x_{i\phi\phi})=\begin{cases}
%0 & x_{i11}+...+x_{i1\phi}+...+x_{ijk}+...+x_{i\phi\phi}=1\\
%-\infty & otherwise
%\end{cases}
%\]
%
%\[
%E_{j}^{2}(x_{1j1},...,x_{1j\phi},...,x_{ijk},...,x_{\phi j\phi})=\begin{cases}
%0 & x_{1j1}+...+x_{1j\phi}+...+x_{ijk}+...+x_{\phi j\phi}=1\\
%-\infty & otherwise
%\end{cases}
%\]
%\[
%E_{k}^{3}(x_{11k},...,x_{1\phi k},...,x_{ijk},...,x_{\phi\phi k})=\begin{cases}
%0 & x_{11k}+...+x_{1\phi k}+...+x_{ijk}+...+x_{\phi\phi k}=1\\
%-\infty & otherwise
%\end{cases}
%\]
%\caption{Definitions of Factor Nodes}
%\label{fig:factor_nodes}
%\end{figure*}

\subsubsection{Optimization Approach}
The general max-sum algorithm iteratively passes messages in a factor graph by Eq.~\eqref{eq:max_sum} and Eq.~\eqref{eq:var_func_maxsum} below. In the equations, $n(C)$ represents neighbor factor nodes of a variable node $C$, and $n(f)$ represents neighbor variable nodes of a factor node $f$. $\mu_{f\rightarrow C}(x)$ defines messages passing from a factor node $f$ to a variable node $C$, while $\mu_{C\rightarrow f}(x)$ defines messages in the reverse direction. Each message is a function of the corresponding variable $x$. For example, if $x$ is a binary variable, there are in total two messages passing in one direction: $\mu_{f\rightarrow C}(0)$ and $\mu_{f\rightarrow C}(1)$.
\begin{eqnarray}
\small
\mu_{f\rightarrow C}(x) 
&=&\max_{x_{1},\cdots,x_{n}}\left(f(x,x_{1},...,x_{n})+ \vphantom{\sum_{C_i\in n(f)\backslash\{C\}}\mu_{C_i\rightarrow f}(x_i)}\right. \nonumber \\
&& \; \; \; \; \; \; \; \; \left.\sum_{C_i\in n(f)\backslash\{C\}}\mu_{C_i\rightarrow f}(x_i)\right) \label{eq:max_sum} \\
\mu_{C\rightarrow f}(x) &=& \sum_{h\in n(C)\backslash\{f\}}\mu_{h\rightarrow C}(x)\label{eq:var_func_maxsum}
\end{eqnarray}

As it stands the solution is computationally expensive. Suppose there are $m$ sources each with an average number of $n$ entities, then in each iteration there are in total $4mn^m$ messages to be updated: only messages between variable nodes and constraint factor nodes need to be updated. There are a total number of $n^m$ variable nodes and each has $m$ constraint factor nodes connecting with it. There are 4 messages between each  factor node and the variable node (2 messages in each direction); so there are in total $4m$ messages for each variable node.

\textbf{Affinity Propagation.} Instead of calculating this many messages, we employ the idea of Affinity Propagation~\cite{Givoni09}. The basic intuition is: instead of passing two messages in each direction, either $\mu_{f\rightarrow C}(x)$ or $\mu_{C\rightarrow f}(x)$, we only need to pass the difference between these two messages, \ie $\mu_{f\rightarrow C}(1) - \mu_{f\rightarrow C}(0)$ and $\mu_{C\rightarrow f}(1) - \mu_{C\rightarrow f}(0)$; at the end of the max-sum algorithm what we really need to know is which configuration of a variable $x$ (0 or 1) yields a larger result.

The concrete calculation---using tri-partite matching to illustrate---is as follows. Let
\begin{eqnarray*}
\beta_{ijk}^{1} &=& \mu_{C_{ijk}\rightarrow E_{i}^{1}}(1)-\mu_{C_{ijk}\rightarrow E_{i}^{1}}(0) \\
\alpha_{ijk}^{1} &=& \mu_{E_{i}^{1}\rightarrow C_{ijk}}(1)-\mu_{E_{i}^{1}\rightarrow C_{ijk}}(0)\enspace.
\end{eqnarray*}
Since
\begin{eqnarray*}
&& \mu_{C_{ijk}\rightarrow E_{i}^{1}}(0) \\
&=&\mu_{E_{j}^{2}\rightarrow C_{ijk}}(0)+\mu_{E_{k}^{3}\rightarrow C_{ijk}}(0)+\mu_{S_{ijk}\rightarrow C_{ijk}}(0) \\
&& \mu_{C_{ijk}\rightarrow E_{i}^{1}}(1) \\
&=& \mu_{E_{j}^{2}\rightarrow C_{ijk}}(1)+\mu_{E_{k}^{3}\rightarrow C_{ijk}}(1)+\mu_{S_{ijk}\rightarrow C_{ijk}}(1) 
\end{eqnarray*}
the difference $\beta_{ijk}^{1}$ between these two messages can be calculated by:
\begin{equation}
\beta_{ijk}^{1}=\alpha_{ijk}^{2}+\alpha_{ijk}^{3}+S_{ijk}(1)\enspace. \label{eq:update_rule_1}
\end{equation}

In the other direction, since
\begin{eqnarray}
\small
&& \mu_{E_{i}^{1}\rightarrow C_{ijk}}(1) \nonumber \\
&=&\max_{x_{i11},...,x_{i\phi\phi}\backslash\{x_{ijk}\}} \left\{ E_{i}^{1}(x_{ijk}=1, x_{i11},...,x_{i\phi\phi}) \vphantom{\sum_{a_{2}a_{3}\ne jk}\mu_{C_{ia_{2}a_{3}}\rightarrow E_{i}^{1}}(x_{ia_{2}a_{3}})} \right. \nonumber\\
 && \; \; +\; \left.\sum_{a_{2}a_{3}\ne jk}\mu_{C_{ia_{2}a_{3}}\rightarrow E_{i}^{1}}(x_{ia_{2}a_{3}})\right\},\nonumber
\end{eqnarray}
in order to get the max value for the message, all the variables except $x_{ijk}$ in $E_{i}^{1}$ should be zero because of the constraint function. Therefore
\begin{eqnarray*}
\mu_{E_{i}^{1}\rightarrow C_{ijk}}(1) &=& \sum_{a_{2}a_{3}\ne jk}\mu_{C_{ia_{2}a_{3}}\rightarrow E_{i}^{1}}(0)
\end{eqnarray*}

Similarly, since
\begin{eqnarray*}
\small
&& \mu_{E_{i}^{1}\rightarrow C_{ijk}}(0)  \\
&=&\max_{x_{i11},...,x_{i\phi\phi}\backslash\{x_{ijk}\}}\left\{ E_{i}^{1}(x_{ijk}=0, x_{i11},...,x_{i\phi\phi}) \vphantom{\sum_{a_{2}a_{3}\ne jk}\mu_{C_{ia_{2}a_{3}}\rightarrow E_{i}^{1}}(x_{ia_{2}a_{3}})} \right. \\
&&\; \; + \; \left. \sum_{a_{2}a_{3}\ne jk}\mu_{C_{ia_{2}a_{3}}\rightarrow E_{i}^{1}}(x_{ia_{2}a_{3}})\right\}\enspace ,
\end{eqnarray*}
to get the max message, one of the variables except $x_{ijk}$ in $E_{i}^{1}$ should be 1, \ie
\begin{eqnarray*}
\mu_{E_{i}^{1}\rightarrow C_{ijk}}(0)&=&\max_{b_{2}b_{3}\ne jk}\left\{ \mu_{C_{ib_{2}b_{3}}\rightarrow E_{i}^{1}}(1) \vphantom{\sum_{a_{2}a_{3}\ne jk,a_{2}a_{3} \ne b_{2}b_{3}}\mu_{C_{ia_{2}a_{3}}\rightarrow E_{i}^{1}}(0)}\right.  \\
&& +  \left.\sum_{a_{2}a_{3}\ne jk,a_{2}a_{3} \ne b_{2}b_{3}}\mu_{C_{ia_{2}a_{3}}\rightarrow E_{i}^{1}}(0)\right\} . 
\end{eqnarray*}
Therefore, subtracting $\mu_{E_{i}^{1}\rightarrow C_{ijk}}(0)$ from $\mu_{E_{i}^{1}\rightarrow C_{ijk}}(1)$ we can get the update formula for $\alpha_{ijk}^{1}$ as follows:
\begin{eqnarray}
\alpha_{ijk}^{1} &=& \min_{b_{2}b_{3}\ne jk}\left\{ -\beta_{ib_{2}b_{3}}^{1}\right\} \enspace. \label{eq:update_rule_2}
\end{eqnarray}

If $i=\phi$, since there is no constraint function $E_{\phi}^{1}$, both $\beta_{ijk}^{1}$ and $\alpha_{ijk}^{1}$ are equal to 0. The update rules Eq.'s~\eqref{eq:update_rule_1} and~\eqref{eq:update_rule_2} can be easily generalized to $m$-partite matching.
\begin{equation}
\beta_{x_1 x_2 \ldots x_m}^{i}=\sum_{j\ne i}\alpha_{x_1 x_2 \ldots x_m}^{j}+S_{x_1 x_2 \ldots x_m}(1) \label{eq:general_rule_1}
\end{equation}
\begin{eqnarray}
&& \alpha_{x_1 x_2 \ldots x_m}^{i} \label{eq:general_rule_2} \\
&=&\min_{\ldots b_{i-1}b_{i+1}\ldots \ne \ldots x_{i-1}x_{i+1}\ldots}\left\{ -\beta_{\ldots b_{i-1}x_ib_{i+1}\ldots}^{i}\right\} \nonumber \enspace.
\end{eqnarray}

\textbf{Achieving exponentially fewer messages.} By message passing with $\alpha, \beta$ instead of $\mu$, the total number of message updates per iteration is reduced by half down to $2mn^m$, which includes $mn^m$ each of $\alpha, \beta$ messages. This number is still exponential in $m$. However, if we look at Eq.~\eqref{eq:update_rule_2} for $\alpha$ messages, the formula actually can be interpreted as Eq.~\eqref{eq:two_values} below, where $j^*k^*$ are minimizers of Eq.~\eqref{eq:update_rule_2}. This shows that for a fixed entity $i$ in source 1, there are actually only two $\alpha$ values for all different combinations of $j$ and $k$, and only the minimizer combination $j^*k^*$ obtains the different value which is the second minimum (we call the second minimum the \emph{exception value} for $\alpha_{ijk}^{1}$, and the minimum as the \emph{normal value}). In other words, there are actually only $2mn$ of the $\alpha$ messages instead of $mn^m$.
\begin{equation}
\alpha_{ijk}^{1}=\begin{cases}
\alpha_{ij^{*}k^{*}}^{1}, & \mbox{for all $jk\ne j^{*}k^{*}$}\\
\textrm{second\_min}\left\{ -\beta_{ib_{2}b_{3}}^{1}\right\}, & \mbox{for $jk=j^{*}k^{*}$}
\end{cases}\enspace. \label{eq:two_values}
\end{equation}

Using this observation, we replace the $\beta$ messages in Eq.~\eqref{eq:update_rule_2} with $\alpha$ messages using Eq.~\eqref{eq:update_rule_1}, to yield:
\begin{eqnarray}
\alpha_{ijk}^{1} &=& -\max_{b_{2}b_{3}\ne jk}\left\{\alpha_{ib_2b_3}^{2}+\alpha_{ib_2b_3}^{3}+S_{ib_2b_3}(1)\right\} .
\label{eq:update_rule_3}
\end{eqnarray}

We can similarly derive $\alpha$ update functions in the other sources. As shown in the following sections, the final resolution only depends on $\alpha$ messages, therefore we only calculate $\alpha$ messages and keep updating them iteratively until convergence or a maximum number of iterations is reached.

\subsubsection{Stepwise Approximation}

In Eq.~\eqref{eq:update_rule_3} we need to find the optimizer, which results in a complexity at least $O(m^2n^{m-1})$ in the general case to compute an $\alpha$ message, because we need to consider all the combinations of $b_2$ and $b_3$ in the similarity function $S_{ib_2b_3}$. To reduce this complexity, we use a stepwise approach (operating like Gibbs Sampling) to find the optimizer. The basic idea is: suppose we want to compute the optimum combination $j^*k^*$ in Eq.~\eqref{eq:update_rule_3}, we first fix $k_1$ and find the optimum $j_1$ which optimizes the right hand side of Eq.~\eqref{eq:update_rule_3} (in a general case, we fix candidate entities in all sources except one). Also, since $k_1$ is fixed, we don't need to consider the similarity $S_{ik_1}$ between $i$ and $k_1$ in $S_{ib_2k_1}(1)$, as shown in Eq.~\ref{eq:gibbs}.
\begin{eqnarray}
j_1&=&\textrm{argmax}_{b_2}\left\{\alpha_{ib_2k_1}^{2}+\alpha_{ib_2k_1}^{3}+S_{ib_2k_1}(1)\right\} \nonumber\\
&=&\textrm{argmax}_{b_2}\left\{\alpha_{ib_2k_1}^{2}+\alpha_{ib_2k_1}^{3}+S_{ib_2}+S_{b_2k_1}\right\}\label{eq:gibbs}
\end{eqnarray}
From Eq.~\ref{eq:gibbs}, we see that this step requires only $O(mn)$ computation, where $m$ is the number of sources ($m=3$ in the current example), because the computation of similarities among the fixed candidate entities need not matter. After we compute $j_1$, we will fix it and compute the optimizer $k_2$ for the third source. We keep this stepwise update continuing until the optimum value does not change any more. Since the optimum value is always increased, we are sure that this stepwise computation will converge. In the general case, the computation is similar, and the complexity is now reduced to $O(mnT)$ where $T$ is the number of steps.

In practice, we want the starting combination to be good enough so that the number of steps needed to converge is small. So, we always select those entities which have high similarities with entity $i$ as the initial combination. Also, to avoid getting stuck in a local optimum, we begin from several different starting points to find the global optimum. We select the top two most similar entities with entity $i$ from each source, and then choose $L$ combinations among them as the starting points. Therefore, the complexity to find the optimum becomes $O(mnTL)$. The second optimum will be chosen among the outputs of each of the $T$ steps within $L$ starting points, requiring another $O(TL)$ computation.

\subsubsection{Final Selection}
We iteratively update all $\alpha$ messages round by round until only a small portion (\eg less than 1\%) of the $\alpha$'s are changing. In the final selection step, the general max-sum algorithm will assign the value of each variable by the following formula.
\begin{eqnarray}
x_{C} &=& \textrm{argmax}_{x}\sum_{h\in n(C)}\mu_{h\rightarrow C}(x) \nonumber
\end{eqnarray}
The difference between the two configurations of $x_{C_{i_1,\ldots,i_m}}$ is:
\begin{eqnarray}
Q(x_{C})&=&\sum_{h\in n(C)}\mu_{h\rightarrow C}(1)-\sum_{h\in n(C)}\mu_{h\rightarrow C}(0) \nonumber \\
&=&\sum_{j=1}^{m}\alpha_{i_1,\ldots,i_m}^{j} + S_{i_1,\ldots,i_m}(1) \nonumber
\end{eqnarray}
This means the final selection only depends on $\alpha$ messages. Therefore, in our algorithm the configuration of the values is equal to
\begin{equation}
x_{i_1,\ldots,i_m}=\begin{cases}
1, & \mbox{$\sum_{j=1}^{m}\alpha_{i_1,\ldots,i_m}^{j} + S_{i_1,\ldots,i_m}(1) \ge 0$} \\
0, & \mbox{otherwise}
\end{cases}\label{eq:selection}
\end{equation}

While we employ heuristics to improve the efficiency of this step
we omit the details as they are technical and follow the general approach laid out above.

\subsubsection{Complexity Analysis}

Suppose there are $m$ sources, each having $n$ entities. The general max-sum algorithm in Eq.'s~\eqref{eq:max_sum} and~\eqref{eq:var_func_maxsum} needs to compute at least $O(mn^m)$ messages per iteration. In addition to the complexity of Eq.~\eqref{eq:max_sum}, each message also requires at least $O(mn^{m-1})$ complexity. The cost is significant when $n\approx 10^6$ and $m>2$.

The complexity of our algorithm can be decomposed into two parts. During message passing, for each iteration we update $O(nm)$ $\alpha$ messages. For each $\alpha$ message, we first need to find the top-$2$ candidates in $m-1$ sources, which requires $O\left((m-1)n\right)$ time. Then, we find the optimum starting from $L$ combinations within the candidates, which requires $O(mnTL)$ time complexity.
%During the final selection stage, Step~\ref{step:selectQ} requires the same complexity as computing an $\alpha$ message.
The sorting of the candidate array requires $O\left(mn\log mn\right)$ time.

A favorable property of message-passing algorithms is that the computation of messages can be parallelized. We leave this for future work.

\subsection{Greedy Algorithm}
The greedy approach to bipartite max-weight matching is simple and efficient, so we explore its natural extension to multi-partite graphs.

\textbf{Algorithm.} The algorithm first sorts all the pairs of entities by their scores, discarding those falling below threshold $\theta$. It then steps through the remaining pairs in order. When a pair $x_i,y_j$ is encountered, with both unmatched so far, the pair entities is matched in the resolution: they form a new \emph{clique} $\{x_i, y_j\}$. If either of $x_i,y_i$ are already matched with other entities and formed cliques, we examine every entity in the two cliques (a single entity is regarded as a singleton clique). If all the entities derive from different sources, then the two cliques are merged into a larger clique with the resulting merged clique included in the resolution. Otherwise, if any two entities from the two cliques are from the same two sources, the current pair $x_i$ and $y_j$ will be disregarded as the merged clique would violate the one-to-one constraint.

%For example, consider the weighted graph shown in Figure~\ref{fig:examples}(a) again. The greedy algorithm will sort all the edges by their weight and then select from the top. In this concrete example, all the blue edges will be ranked at the top and will be all added into the resolution result because they don't violate the one-to-one constraint. When we want to add more matching pairs from the red edges, neither of them can be added because of the one-to-one constraint. So the result of the greedy algorithm will be the same as the one shown in Figure~\ref{fig:examples}(b).

\textbf{Complexity analysis.} The time complexity of the sorting step is $O(n^2m^2(\log{n}+\log{m}))$, where $m$ is the number of sources and $n$ is the average number of entities in each source. The time complexity of the selection step is $O(n^2m^3)$ as we must examine $O(n^2m^2)$ edges with each examination involving $O(m)$ time to see if adding this edge will violate the one-to-one constraint. In practice the greedy approach is extremely fast.

\textbf{Approximation guarantee.} In addition to being fast and easy to implement, the bipartite greedy algorithm is well-known to be a 2-approximation of exact max-weight matching on bipartite graphs. We generalize this competitive ratio to our new multi-partite greedy algorithm via duality.

\begin{theorem}
\label{thm:np-cr} On any weighted $m$-partite graph, the total weight
of greedy max-weight multi-partite graph matching is at least half
the max-weight matching.
\end{theorem}

\begin{proof}
The primal program for multipartite exact max-weight matching is below.
Note we have pairs of sources indexed by $u,v$ with
corresponding inter-source edge set $\mathcal{E}_{uv}$ connecting nodes of
$D_{u}$ and $D_{v}$.
\begin{eqnarray}
p^{\star}=\max_{\mathbf{\Delta}} &  & \sum_{uv}\sum_{ij\in \mathcal{E}_{uv}}s_{ij}\delta_{ij}\label{eq:np-primal}\\
\st &  & \sum_{j\,:\,(ij)\in \mathcal{E}_{uv}}\delta_{ij}\leq1,\;\forall u,v,\forall i\in D_{u}\nonumber \\
 &  & \mathbf{\Delta}\succeq\mathbf{0}\nonumber \end{eqnarray}
The primary variables select the edges of the matching, the objective function measures the total selected weight, and the constraints enforce the one-to-one property and that selections can only be made or not (and not be negative; and in sum, ensuring the optimal solution rests on valid Boolean values).

The Lagrangian function corresponds to the following, where we introduce dual variables per constraint, so as to bring constraints into the objective thereby forming an unconstrained optimization:
\begin{eqnarray*}
&& \mathcal{L}(\mathbf{\Delta},\mathbf{\lambda},\mathbf{\nu}) \\
& = & \sum_{uv}\sum_{ij\in \mathcal{E}_{uv}}s_{ij}\delta_{ij} \\
&& +\; \sum_{uv}\sum_{i\in D_{u}}\lambda_{uvi}\left(1-\sum_{j\,:\,(ij)\in \mathcal{E}_{uv}}\delta_{ij}\right) \\
&& +\; \sum_{uv}\sum_{ij\in \mathcal{E}_{uv}}\nu_{ij}\delta_{ij}\;.
\end{eqnarray*}
Inequality constraints require non-negative duals so that constraint violations are penalized appropriately.

The Lagrangian dual function then corresponds to maximizing the Lagrangian over the primal variables:
\begin{eqnarray*}
&& g(\mathbf{\lambda},\mathbf{\nu}) \\
&=& \sup_{\mathbf{\Delta}}\mathcal{L}(\mathbf{\Delta},\mathbf{\lambda},\mathbf{\nu})\\
 & = & \sum_{uv}\sum_{i\in D_{u}}\mathbf{\lambda}_{uvi}+\sup_{\mathbf{\Delta}}\sum_{uv}\sum_{ij\in \mathcal{E}_{uv}}\left(u_{ij}+\nu_{ij}-\lambda_{uvi}-\lambda_{uvj}\right)\delta_{ij}\\
 & = & \begin{cases}
\mathbf{\lambda}\cdot\mathbf{1}\;, & \mbox{if }s_{ij}+\nu_{ij}-\lambda_{i}-\lambda_{j}\leq0,\, \forall uv;\; i,j\in \mathcal{E}_{uv}\\
\infty\;, & \mbox{otherwise}\end{cases}
\end{eqnarray*}

Next we form the dual LP, which minimizes the dual function subject to the non-negative dual variable constraints. This is the dual LP to our original primal. Dropping the superfluous dual variables:
\begin{eqnarray}
\min_{\mathbf{\lambda}} &  & \sum_{uv}\sum_{i\in D_u}\lambda_{uvi}\label{eq:bp-dual-1}\\
\st &  & s_{ij}\leq\lambda_{uvi}+\lambda_{uvj}\;\forall uv\forall ij\in \mathcal{E}_{uv}\nonumber \\
 &  & \mathbf{\lambda}\succeq\mathbf{0}\nonumber \end{eqnarray}

Finally we may note that the greedily matched weights form a feasible
solution for the dual LP. Moreover
the value of the dual LP under this feasible solution is twice
the greedy total weight. By weak duality this feasible solution bounds
the optimal primal value proving the result.
\end{proof}

The same argument as is used in the bipartite case demonstrates that this
bound is in fact sharp: there exist pathological weighted graphs for
which greedy only achieves half the max weight. In real-world matching
graphs are sparse and their weights are relatively well-behaved, and so
in practice greedy usually achieves much better than this worst-case
approximation as demonstrated in our experiments.

\section{Experimental Methodology}\label{sec:experimental}

We present our experimental results on both real-world and synthetic data. Our movie dataset aims to demonstrate the performance of our algorithms in a real, very large-scale application; the publication dataset adds further support to our conclusions on additional real-world data; while the synthetic data is used to stress test multi-partite matching in more difficult settings.

\subsection{Movie Datasets}
For our main real-world data, we obtained movie meta-data from six popular online \singleDoubleBlind{movie sources used in the Bing movie vertical (\cf Figure~\ref{fig:bing})}{movie sources} via a combination of site crawling and API access. All sources have a large number of entities: as shown in Table~\ref{table:movie_data}, the smallest has over 12,000, while the largest has over half a million. As noted in Section~\ref{sec:intro}, while we only present results on one real-world dataset, the data is three orders of magnitude larger than typical for real-world benchmarks in the literature and is from a real production system~\cite{KopckeTR10}.

\begin{table}[t]
\centering
\caption{Experimental movies data}
\label{table:movie_data}
\begin{tabular}{|c|c|c|c|c|c|}
\hline
Source & Entities & Source & Entities & Source & Entities\\
\hline \hline
AMG & 305,743 & Flixster & 140,881 & IMDB & 526,570 \\
\hline
iTunes & 12,571 & MSN & 104,385 & Netflix & 75,521 \\
\hline
\end{tabular}
\end{table}

Every movie obtained had a unique ID within each source and a title attribute. However, other attributes were not universal. Some sources tended to have more metadata than others. For example, IMDB generally has more metadata than Netflix, with more alternative titles, alternative runtimes, and a more comprehensive cast list. No smaller source was a strict subset of a larger one. For example, there are thousands of movies on Netflix that cannot be found on AMG, even though AMG is much larger. We base similarity measurement on the following feature scores:
\begin{itemize}\itemsep0.005in
\item Title: Exact match yields a perfect score. Otherwise, the fraction of normalized words that are in common (slightly discounted). The best score among all available titles is used.
\item Release year: the absolute difference in release years up to a maximum of 30.
\item Runtime: the absolute difference in runtime, up to a maximum of 60 minutes.
\item Cast: a count of the number of matching cast members up to a maximum of five. Matching only on part of a name yields a fractional score.
\item Directors: names are compared like cast. However, the number matching is divided by the length of the shorter director list.
\end{itemize}
Although the feature-level scores could be improved---\eg by weighing title words by TF-IDF, by performing inexact words matches, by understanding that the omission of ``part 4" may not matter while the omission of ``bonus material" is significant---our focus is on entity matching using given scores, and these functions are adequate for obtaining high accuracy, as will be seen later.

After scoring features, we used regularized logistic regression to learn weights to combine feature scores into a single score for each entity pair. In order to train the logistic regression model and also evaluate our entity matching algorithms, we gathered human-labeled truth sets of truly matching movies. For each pair of movie sources, we randomly sampled hundreds of movies from one source and then asked human judges to find matching movies from the other source. If there exists a matching, then the two movies will be labeled as a matching pair in the truth set. Otherwise, the movie will be labeled as non-matching from the other source, and any algorithm that assigns a matching to that movie from the other source will incur a false positive. We also employ the standard techniques of blocking~\cite{Kopcke2010} in order to avoid scoring all possible pairs of movies. In the movie datasets, two movies belong to the same block if they share at least one normalized non-stopword in their titles.

\subsection{Publications Datasets}
\label{sec:pubdata}

\begin{table}[t]
\centering
\caption{Experimental publications data}
\label{table:pubs_data}
\begin{tabular}{|c||c|c|c|}
\hline
Source & DBLP & ACM & Scholar \\
\hline
Entities & 2,615 & 2,290 & 60,292 \\
\hline
\end{tabular}
\end{table}

We follow the same experimental protocol on a second real-world publications dataset. The data collects
publication records from three sources as detailed in Table~\ref{table:pubs_data}. Each record has the title,
authors, venue and year for one publication.

\subsection{Performance Metrics on Real-World Data}
We evaluate the performance of both the Greedy and the Message Passing algorithms via precision and recall. For any pair of sources, suppose $\hat{R}$ is the output matching from our algorithms. Let $R_{+}$ and $R_{-}$ denote the matchings and non-matchings in our truth data set for these two sources, respectively. Then, we calculate the following standard statistics:
\begin{eqnarray*}
TP &=& |\hat{R} \cap R_{+}|\\
FN &=& |R_{+} \backslash \hat{R}| \\
FP &=& |\hat{R} \cap R_{-}| \\
&&\ +\ \left|\left\{\left.(x,y)\in \hat{R} \right| \exists z\neq y, (x,z)\in R_{+}\right\}\right| \\
&&\ +\ \left|\left\{\left.(x,y)\in \hat{R} \right| \exists z\neq x, (z,y)\in R_{+}\right\}\right|\enspace.
\end{eqnarray*}

The precision is then calculated as $\frac{TP}{TP+FP}$, and recall is calculated as $\frac{TP}{TP+FN}$. By varying threshold $\theta$
we may produce a sequence of precision-recall pairs producing a PR curve.

\subsection{Synthetic Data Generation}
To thoroughly compare our two approaches to general MPEM problems, we design and carry out several experiments on synthetic data with different difficulty levels. This data is carefully produced by Algorithm~\ref{alg:synthetic_gen} to simulate the general entity matching problem.

\begin{algorithm}[t!]
\caption{Synthetic Data Generation \label{alg:synthetic_gen}}
\begin{algorithmic}[1]
\REQUIRE Generate $n$ entities in $m$ sources; $k$ features per entity
\STATE \textbf{for} $i=1$ to $n$ \textbf{do}
\STATE \hspace{1em} for the current entity, generate $k$ values \\
	\hspace{1em} $F^1, F^2, \ldots, F^k\sim Unif[0,1]$ as true features;
\STATE \hspace{1em} \textbf{for} $j=1$ to $m$ \textbf{do}
\STATE \hspace{2em} for each source, generate feature values for \\ \hspace{2em} the current entity as follows:
\STATE \hspace{2em} \textbf{for} $t=1$ to $k$ \textbf{do}
\STATE \hspace{3em} each feature value $f^t$ in $f_{j:i}$ sampled \\ \hspace{3em} from $\mathcal{N}(F^t, \sigma^2)$
\STATE \textbf{done}
\STATE Compute scores between entities across sources:
\STATE \textbf{for} $i_1=1$ to $m-1$ \textbf{do}
\STATE \hspace{1em} \textbf{for} $i_2=i_1+1$ to $m$ \textbf{do}
\STATE \hspace{2em} \textbf{for} $j_1=1$ to $n$ \textbf{do}
\STATE \hspace{3em} \textbf{for} $j_2=1$ to $n$ \textbf{do}
\STATE \hspace{4em} score entities $j_1, j_2$ from sources $i_1, i_2$ \\
	\hspace{4em} using values $f_{i_1:j_1}, f_{i_2:j_2}$
\STATE \textbf{done}
\end{algorithmic}
\end{algorithm}

We randomly generate features of entities instead of randomly generating scores of entity pairs. This corresponds to the real world, where each source may add noise to the true meta-data about an entity. In the synthetic case, we start with given true data and add random noise. An important variable in the algorithm is the variance $\sigma^2$. As $\sigma$ increases, even the same entity may have very different feature values in different sources, so the entity matching problem gets harder. In our experiment, we vary the value of $\sigma$ through 0.02, 0.04, 0.06, 0.08, 0.1 and 0.2 to create data sets with different difficulty levels.
%We cannot generate difficult ER problems by simply adding noise to feature scores directly, as this will produce very unnatural weights graphs: entities $a$, $b$, $c$ across three sources may enjoy high similarity pre-noise, but independent noise may move $a$ further from $b$ but closer to $c$ while leaving $b$ and $c$ alone.

We set the number of entities $n=100$, the number of sources $m=3$, and the number of features $k=5$. We use the normalized inner product of two entities as their similarity score.
%which is defined as follows:
%\begin{equation}
%sim(i_1:j_1, i_2:j_2)=\frac{f_{i_1:j_1}*f_{i_2:j_2}}{||f_{i_1:j_1}||\cdot||f_{i_2:j_2}||} \nonumber
%\end{equation}

\subsection{Performance Metrics on Synthetic Data}
Since all the entities and sources are symmetric, we evaluate the performance of different approaches on the whole data set. For a data set with $m$ sources and $n$ entities, there are a total of $\frac{nm(m-1)}{2}$ true matchings. If the algorithm outputs $T$ entity matchings out of which $C$ entity pairs are correct, then the precision and recall are computed as:
\[Precision = \frac{C}{T},\,\,\,Recall = \frac{2C}{nm(m-1)} \nonumber\]
As another measure we also use F-1, the harmonic mean of precision and recall, defined as $\frac{2*precision*recall}{precision+recall}$.

\subsection{Unconstrained ManyMany Baseline}

As a baseline approach, we employ a generic unconstrained ER approach---denoted \AlgMatchMM---that is common to many previous works~\cite{KDD03-MARLIN,SIGMOD09-logistic,QDBMUD08-STEM,Kopcke2010,KDD98-logistic,ActiveLearn-KDD02,JCoopInfoSys,ActiveAtlas-KDD02,InfoSys-MultClassSys}. Given a score function $f: D_i \times D_j \to \mathbb{R}$ on sources $i,j$ and tunable threshold $\theta\in\mathbb{R}$,  \AlgMatchMM\ simply resolves all tuples with score exceeding the threshold. Compared to our message passing and greedy algorithms described in Section~\ref{sec:algorithms}, \AlgMatchMM\ allows an entity in one data source to be matched with multiple entities in another data source. Because of this property, it is also important to notice that adding more data sources will not help improve the resolution results of \AlgMatchMM\ on any pair of sources, \ie the resolution results on a pair of two sources using \AlgMatchMM\ are invariant to any other sources being matched.

\section{Results}\label{sec:results}
We now present results of our experiments.

\subsection{Results on Movie Data}
The real-world movie dataset contains millions of movie entities, on which both message-passing and greedy approaches are sufficiently scalable to operate. This section's experiments are designed primarily for answering the following:
\begin{enumerate}\itemsep0.005in
\item Whether multi-partite matching can get better precision and recall compared with simple bipartite matching---do additional sources improve statistical performance;
\item Whether the one-to-one global constraint improves the resolution results in real movie data;
\item Whether the message-passing-based approach achieves higher total weight than the heuristic greedy approach; and
\item How the two approaches perform on multi-partite matching of real-world movie data.
\end{enumerate}

\subsubsection{Multipartite vs. Bipartite}\label{sec:pr_curve}
We examine the effectiveness of multi-partite matching by adding sources to the basic bipartite matching problem. Specifically, we use two movie sources (msnmovie and flixster in our experiments) as the target source pair, and perform bipartite matching on these two sources to obtain baseline performance. We then add another movie source (IMDB) and perform multi-partite matching on these three sources, and compare the matching results on the target two sources against the baseline. Finally, we perform multi-partite matching on all six movie sources and record the results on the target two sources.

\begin{figure}[t]
\begin{center}
\begin{minipage}[t]{1\columnwidth}
\centering
\includegraphics[width=0.9\linewidth]{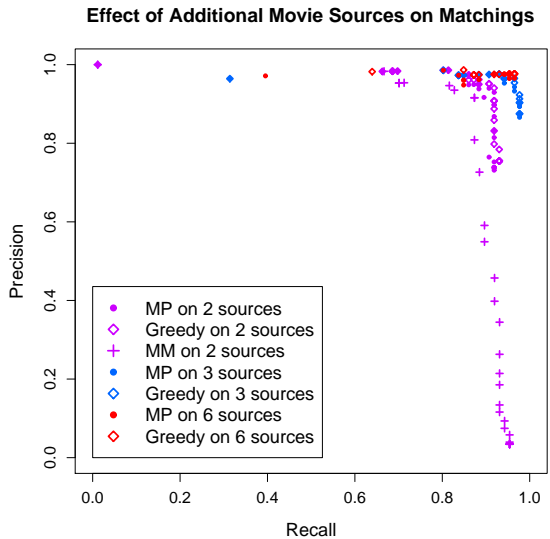}
\caption{Performance of resolution on increasingly many movie
sources. The algorithms enjoy comparably high performance which improves
with additional sources.}
\label{fig:movies-pr}
\end{minipage}
\end{center}
\end{figure}

Seven groups of results are shown together in Figure~\ref{fig:movies-pr}. We use diamonds to plot the results from the greedy approach, dots to plot the results from the message-passing approach, and plus symbols to plot the results from \AlgMatchMM. The results for different numbers of sources are recorded with different colors. The PR curve is generated by using different thresholds $\theta$ on scores. Specifically, given a threshold, we regard all the entity pairs which have similarity below the threshold as having 0 similarity and as never being matched. Then, we perform entity resolution based on the new similarity matrix and plot the precision and recall value as a point in the PR curve. We range from threshold 0.51 through 0.96 incrementing by 0.03.

From the results, we can first see that when comparing the three approaches on two sources, both message passing and greedy are much better than \AlgMatchMM. This demonstrates the effectiveness of adding the global one-to-one constraint in entity resolution in real movie data. It's also important to notice that without the one-to-one constraint, \AlgMatchMM\ is not affected at all when new sources are added; while for the approaches with the global one-to-one constraint, the resolution improves significantly with increasing number of sources, particularly going from 2 to 3. This further shows the importance of having the one-to-one constraint in multi-partite entity resolution: it facilities appropriate linking of the individual pairwise matching problems.

In addition, for both message-passing and greedy approaches, the PR curve of using three sources for matching is higher than using only two sources for matching. This means that with an additional source IMDB, which has the largest number of entities among all the sources, the resolution of msnmovie and flixster is much improved. The matching results on all of the six sources are also higher than using only two sources, but only slightly higher than on three sources. This implies that the other three movie sources do not provide much more information beyond IMDB, likely because IMDB is a very broad source for movies and it's already very informative with good quality metadata to help match msnmovie and netflix.

For any number of sources, the precision and recall of greedy is very similar to that of message passing and very close to 1.0. This may be due to our feature scores and similarity functions being already pretty good for movie data. In order to determine whether they are fundamentally similar, or whether greedy was benefiting from this data
set having low feature noise, we designed the synthetic experiment which is described in the following sections. Since the experiments on real movie data have already shown the advantage of these two approaches over \AlgMatchMM, we do not compare \AlgMatchMM\ in the synthetic experiment.

%We can show another case for comparison, say netflix and itunes. That would add another 3*2*2=12 results

\begin{figure*}[t]
\begin{center}
\begin{minipage}[t]{0.32\textwidth}
\centering
\includegraphics[width=1\linewidth]{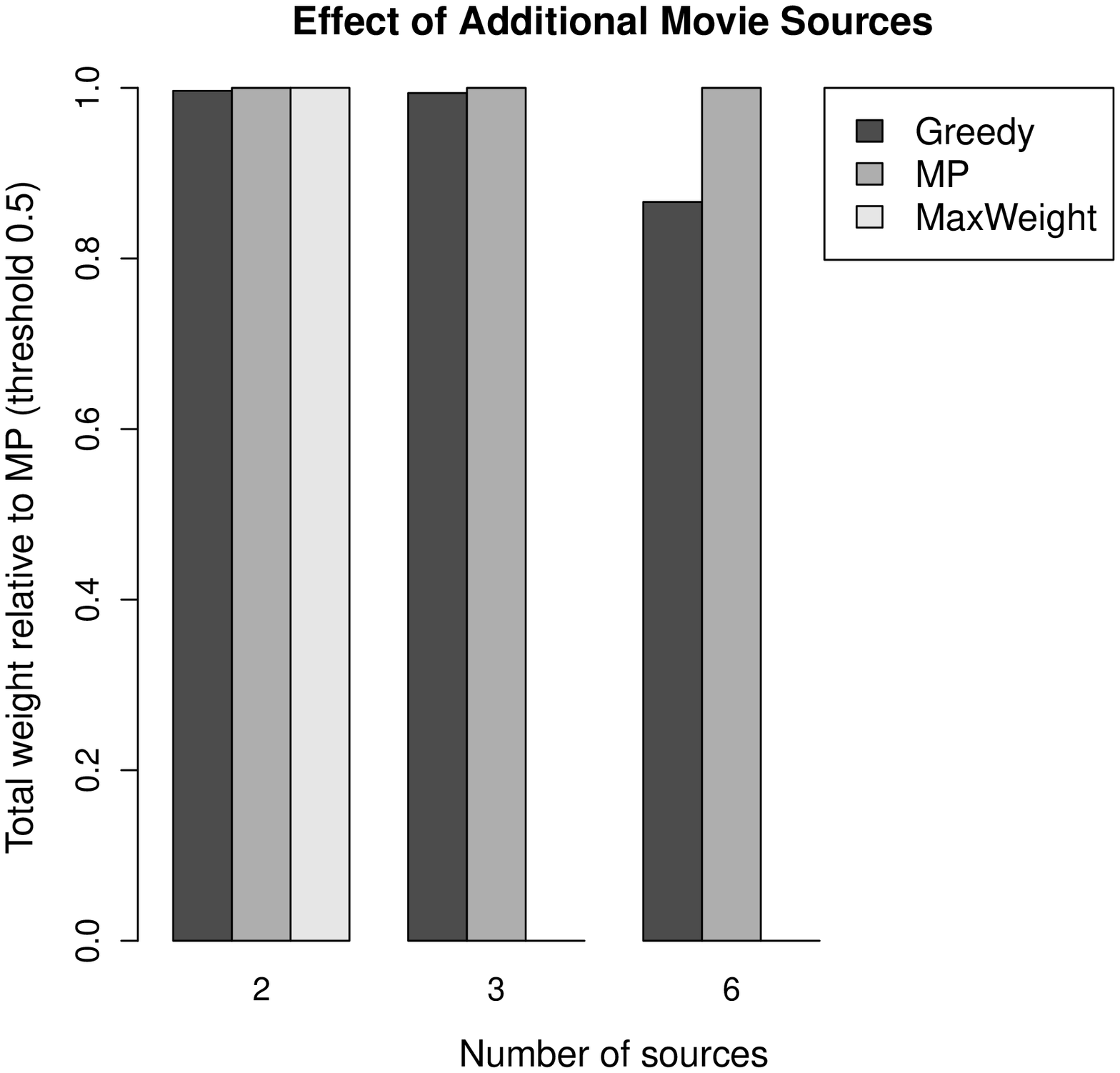}
\caption{Comparing weights, with true max weight for 2 sources. To
compare across varying number of sources, we plot weights
relative to Message Passing.}
\label{fig:movies-weight}
\end{minipage}\hfill
\begin{minipage}[t]{0.32\textwidth}
\centering
\includegraphics[width=1\linewidth]{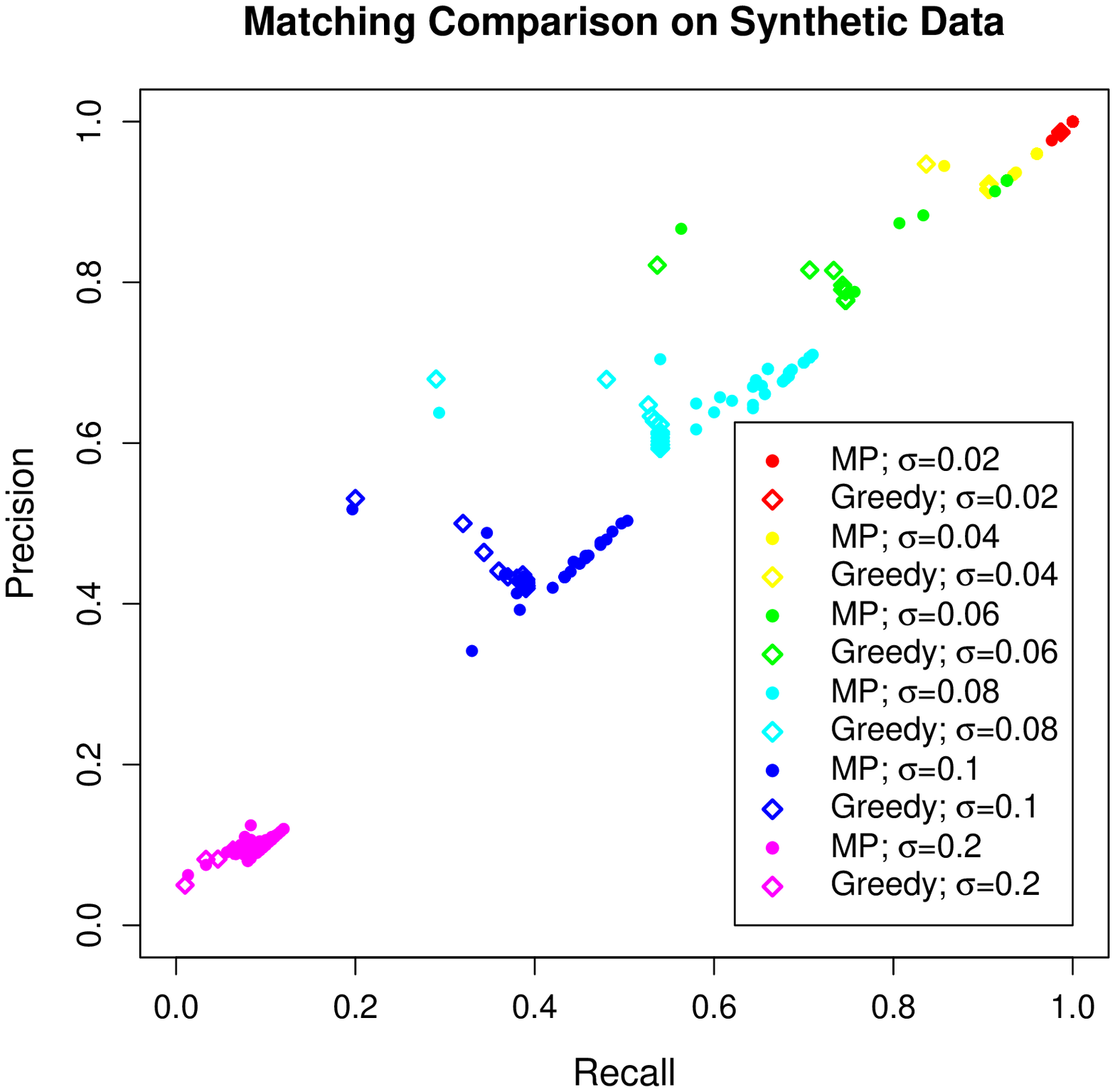}
\caption{Precision-Recall comparison of Message Passing and Greedy matching
on synthetic data with varying amounts of score noise.}
\label{fig:synthetic-pr}
\end{minipage}\hfill
\begin{minipage}[t]{0.32\textwidth}
\centering
\includegraphics[width=1\linewidth]{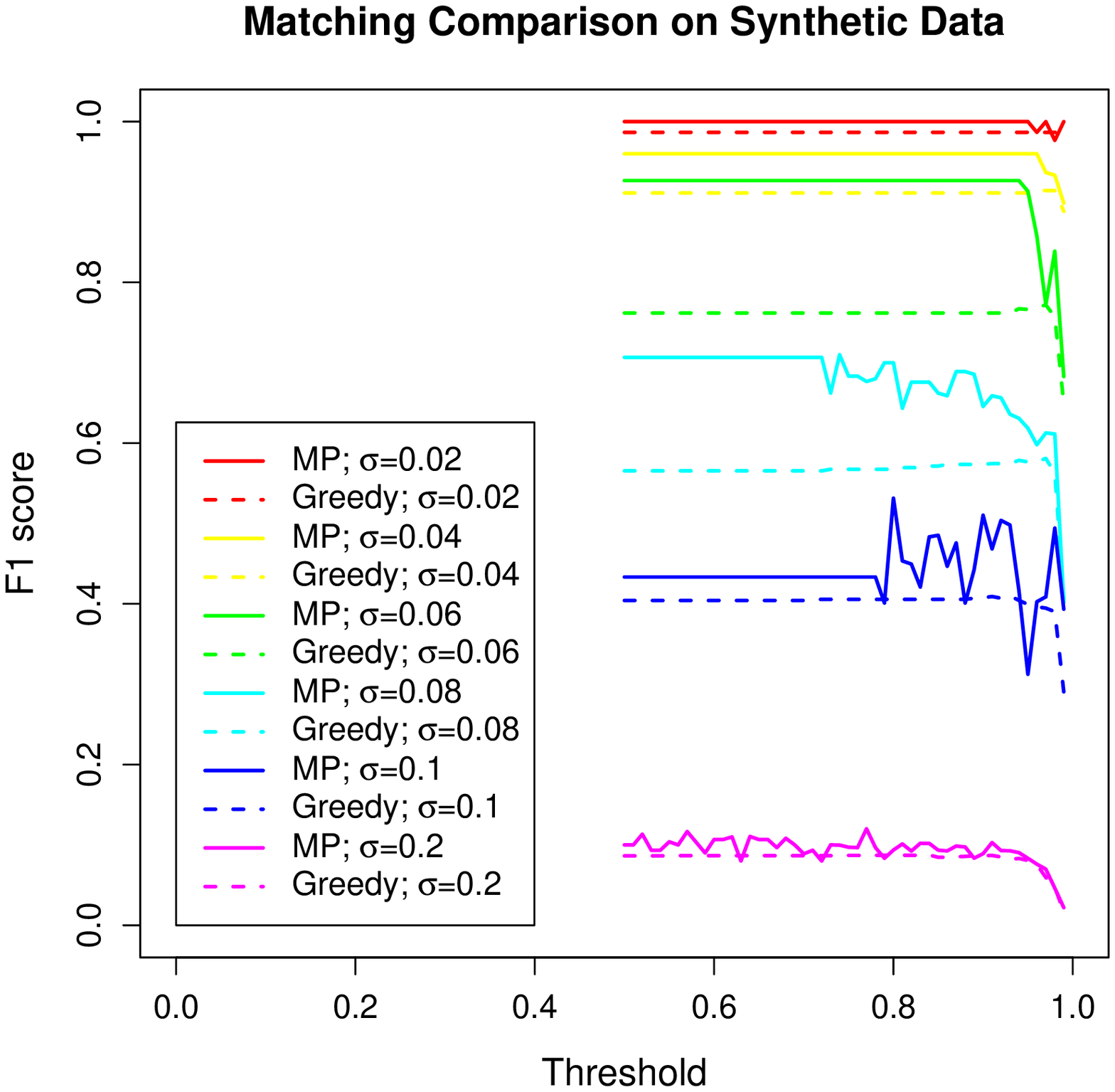}
\caption{F1 scores of Message Passing and Greedy matching
on synthetic data with varying amounts of score noise.}
\label{fig:synthetic-f1}
\end{minipage}
\end{center}
\end{figure*}

\subsubsection{Total weight}
Next we compare the total weight of the matchings output by the message-passing and greedy approaches. The total weight of a matching is calculated as the sum of similarity scores of each matching pair. For example, if a movie $x$ from IMDB matches with movie $y$ in msnmovie and movie $z$ in netflix, then the total weight is $sim(x,y)+sim(x,z)+sim(y,z)$. Since the message-passing approach goes to additional effort to better-approximate the maximum weight matching, we expect that its total weight should be higher than greedy.

%\begin{figure}[t]
%\begin{center}
%\begin{minipage}[t]{1\columnwidth}
%\centering
%\includegraphics[width=1\linewidth]{figures/analysis04.rel.besides.eps}
%\caption{Comparing total weight achieved by Message Passing and Greedy
%on multiple sources, with the true maximum weight for 2 sources. In order to
%compare the methods across varying number of sources, we plot total weights
%relative to Message Passing.}
%\label{fig:movies-weight}
%\end{minipage}
%\end{center}
%\end{figure}

In Figure~\ref{fig:movies-weight}, we show the comparative results when performing multi-partite matching on two, three, and six sources, where the threshold is set as 0.51. For two sources, since exact maximum-weight bipartite matching algorithms is tractable, we compute the true maximum weight and display as ``MaxWeight''. For comparison, we use the weight of the message-passing approach as the reference, and plot the relative value of the greedy and max-weight approaches to this reference. From the figure, we can see that on all the datasets, the message-passing approach gets higher total weight than the greedy approach. Also, as the number of sources increases, the difference between the total weight of the message and greedy approaches gets larger. For example, on six sources the total weight of the message-passing approach is more than 10\% higher than greedy. In addition, in the two source comparison we can see that the message-passing approach (total weight: 76404) gets almost the same total weight as the maximum-weight matching (total weight: 76409), while the greedy approach's total weight is a little lower (76141). This suggests that our message-passing approach approximates the maximum-weight matching well. On the other hand, the result also implies that the movie matching data is far from pathological since the total weight for greedy is very close to the maximum-weight matching on two sources, which is in stark contrast with the 2-approximation worst case.

%\textbf{In total, we can present (1 plus 1 plus 1) by 2(thresholds) results here}

%We show the PR and F-1 curves of the two approaches on the movie matching problem. Basically, we use matchings between two-sources, three-sources, and six-sources for comparison. We randomly select two example pairs from two-source matching, two example pairs from three-source matching, and two example pairs from six-source matching for evaluation.
%
%\textbf{In total, we can present (2+2+2) by 2 (approach) by 2 (metric) results here}

\subsection{Results on Publication Data}

We further explore our approaches with the publication dataset of Section~\ref{sec:pubdata}, for which
results are shown in Figure~\ref{fig:dblp}. We experiment on matching DBLP vs. Scholar, with matching DBLP vs. ACM vs. Scholar, comparing the results of the bipartite and tripartite matchings with truth labels on DBLP-Scholar. The results lead similar conclusions drawn from the movie data: (1) globally-constrained matching with an additional source improve the accuracy of matching DBLP-Scholar alone; and (2) one-to-one constrained approaches perform better than the unconstrained \AlgMatchMM approach.

\begin{figure}[t]
\begin{center}
\begin{minipage}[t]{1\columnwidth}
\centering
\includegraphics[width=0.85\linewidth]{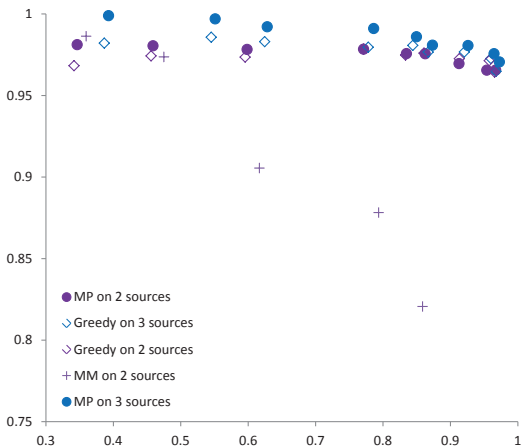}
\caption{Performance of resolution on increasingly many publication
sources.}
\label{fig:dblp}
\end{minipage}
\end{center}
\end{figure}

\subsection{Results on Synthetic Data}

As discussed in Section~\ref{sec:experimental}, we created synthetic datasets with different levels of feature noise by varying parameter $\sigma$. Our primary aim with the synthetic data is to see if greedy is truly competitive to message passing, or if it only achieves similar precision/recall performance in the case of low feature noise on the relatively well-curated movie data.
Our experiments are on 3 sources with 100 entities. We use precision, recall and F-1 score as the evaluation metrics to compare the two approaches.

%
%\subsection{Multipartite vs. Bipartite}
%We still compare matchings between two-sources, three-sources, and six-sources. We use the average results from 5 runs. We can set 5 different difficulty levels. For each level, we compare the evaluation results on two target sources between using 2, 3, and 6 sources.
%
%\textbf{In total, we can present 3 (datasets) by 5 (levels) by 2 (metric) by 2 (approach) results here}

%\begin{figure*}[t]
%\begin{center}
%\begin{minipage}[t]{0.49\textwidth}
%\centering
%\includegraphics[width=1\linewidth]{figures/analysis05.pr.eps}
%\caption{Precision-Recall comparison of Message Passing and Greedy matching
%on synthetic data with varying amounts of score noise.}
%\label{fig:synthetic-pr}
%\end{minipage}\hfill
%\begin{minipage}[t]{0.49\textwidth}
%\centering
%\includegraphics[width=1\linewidth]{figures/analysis05.f1.eps}
%\caption{F1 scores of Message Passing and Greedy matching
%on synthetic data with varying amounts of score noise.}
%\label{fig:synthetic-f1}
%\end{minipage}
%\end{center}
%\end{figure*}

Figure~\ref{fig:synthetic-pr} shows the PR curves of the two approaches on the different data sets. Here, $\sigma=0.02$ represents the easiest dataset (\ie generated with minimal noise) and $\sigma=0.2$ represents the most noisy dataset (\ie generated with significant noise). We use a different color to present the results on different data quality and use different point markers to present results from the different approaches. From the figure, we see that for the easiest dataset ($\sigma=0.02$), both greedy and message passing work exceptionally well, reminiscent of the movie matching problem. When the dataset gets more noisy, both approaches experience degraded results, but the decrease of the message-passing approach is far less than the greedy approach. As shown in the figure, for datasets with $\sigma=0.04, 0.06, 0.08, 0.1$, message-passing operates far better than greedy. Finally, when the data arrives at a very noisy level ($\sigma=0.2$), both approaches perform equally poorly.

In Figure~\ref{fig:synthetic-f1}, we show the F-1 values of the two approaches along with different thresholds. We can see very clearly the contrast between the two approaches when the data becomes increasingly noisy. Here colors again denote varying noise level, while the line type denotes message passing (solid) or greedy (dashed).
For example, the gap between the solid and dashed lines is very small at the top of the figure when the data is relatively clean. However, as $\sigma$ increases the gap too increases, and reaches an apex on the green curves ($\sigma=0.06$). Later, as $\sigma$ increases even more, the gap becomes smaller but still exists. At last, when $\sigma=0.2$, the gap becomes very small, which means the two approaches perform almost the same under severe noise. This is to be expected: no method can perform well under extreme noise.

In sum, from the experimental results on the synthetic data, we can conclude that for the datasets with the least feature noise, both the greedy approach and the message-passing approach perform very well, while for the most noisy datasets, both of these two approaches perform poorly. On the other hand, when the feature noise is between these two extremes, message-passing is much more robust than greedy.

\subsection{Convergence and Scalability}
The complexity of our message-passing approach depends on the number of iterations that the algorithm needs to converge. In this section, we examine the convergence of the algorithm and empirically compare the efficiency of the message-passing approach with the greedy approach.

In the convergence experiment, we use 6 data sources each having 1k entities making 6k $\alpha$ messages total. We change the threshold (as in Sec.~\ref{sec:pr_curve}) from 0.3 to 0.8 with increment 0.1 to filter the candidate matchings. A higher value of the threshold will have a lower number of edges left in the weighted graph, so the total weight of the final matching will also be lower and the iterations needed to converge is also smaller. In Figure~\ref{fig:convergence_test}(a), we show how many iterations the message-passing approach needs to converge with different thresholds, and Figure~\ref{fig:convergence_test}(b) shows the total weight after each iteration of message passing. Both of these two graphs show that with 6 sources and 1,000 entities, the message-passing approach converges quickly. With all the different thresholds, the approach converges within 50 iterations. This indicates that the factor $T$ in the complexity analysis for the message-passing approach is much smaller than the total number of entities $n$.

\begin{figure}[t]
\centering
\begin{tabular}{cc}
  \includegraphics[width=4cm]{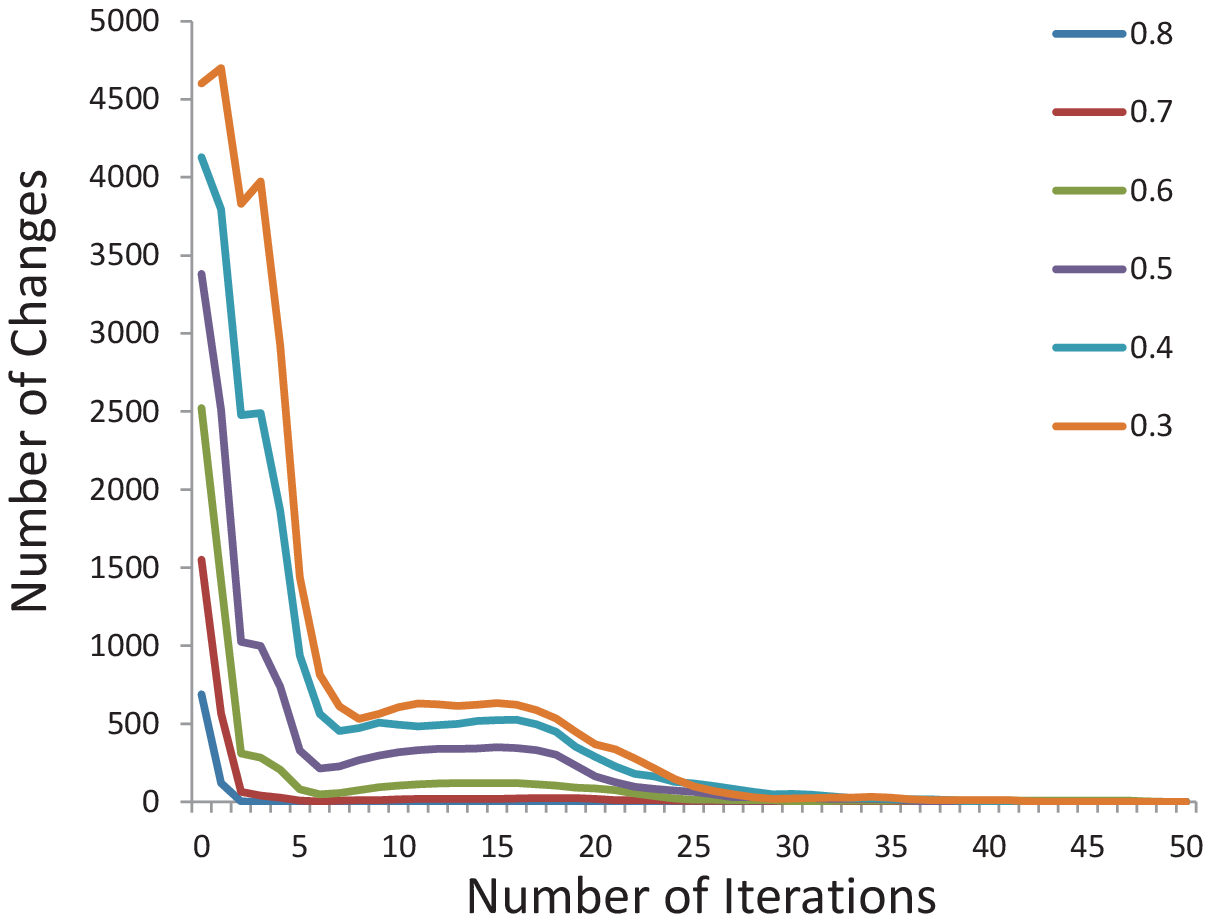} &
  \includegraphics[width=4cm]{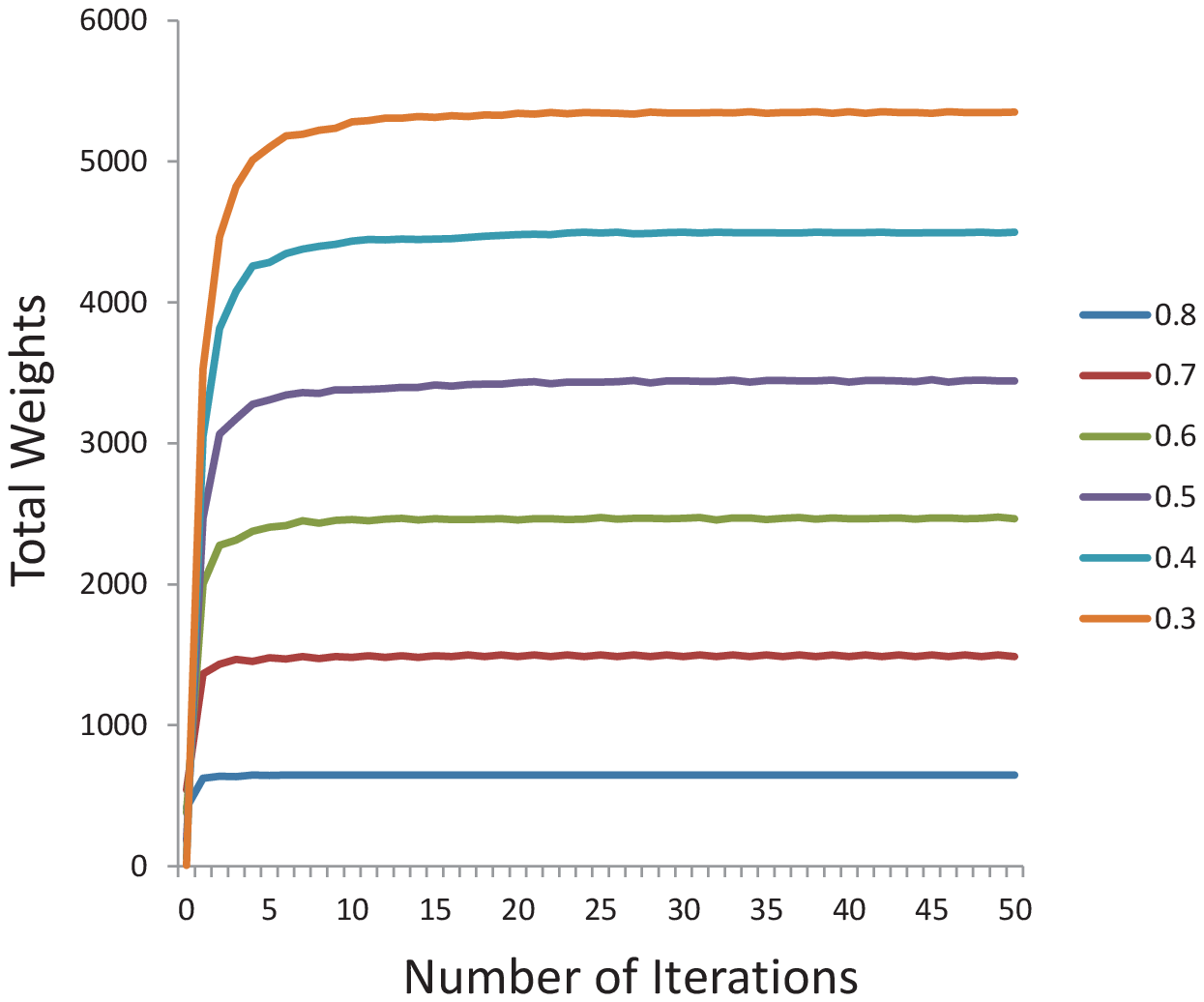} \\
  (a) Change of $\alpha$ messages & (b) Change of total weight
\end{tabular}
\caption{Convergence Test}
\label{fig:convergence_test}
\end{figure}

We also conducted an experiment to empirically compare the efficiency of the message-passing approach and the greedy approach. There is no doubt that the message-passing approach is much slower than the greedy approach, but the purpose of the experiment is to see how well the message-passing approach scales when the number of entities and the number of sources increase. The experiment was performed on a computer with Intel Core i5 CPU M560@2.67GHz and 4GB Memory. Figure~\ref{fig:running_comparison}(a) shows the time cost comparison between the two approaches when we fix the number of sources as 6 and increase the number of entities per source from 200 to 1,000, and in Figure~\ref{fig:running_comparison}(b) we fix the number of entities per source as 600 and increase the number of sources 3 to 7. We can see that the greedy approach scales well. While the time cost of the message-passing approach increases much faster, the time complexity of the algorithm is still acceptable for solving real problems. In practice, it takes 1 to 2 hours to use message passing to find the true matching on the real, large-scale movie dataset, which has 6 sources, some of which containing around half a million movie entities.

\begin{figure}[t]
\centering
\begin{tabular}{cc}
  \includegraphics[width=4cm]{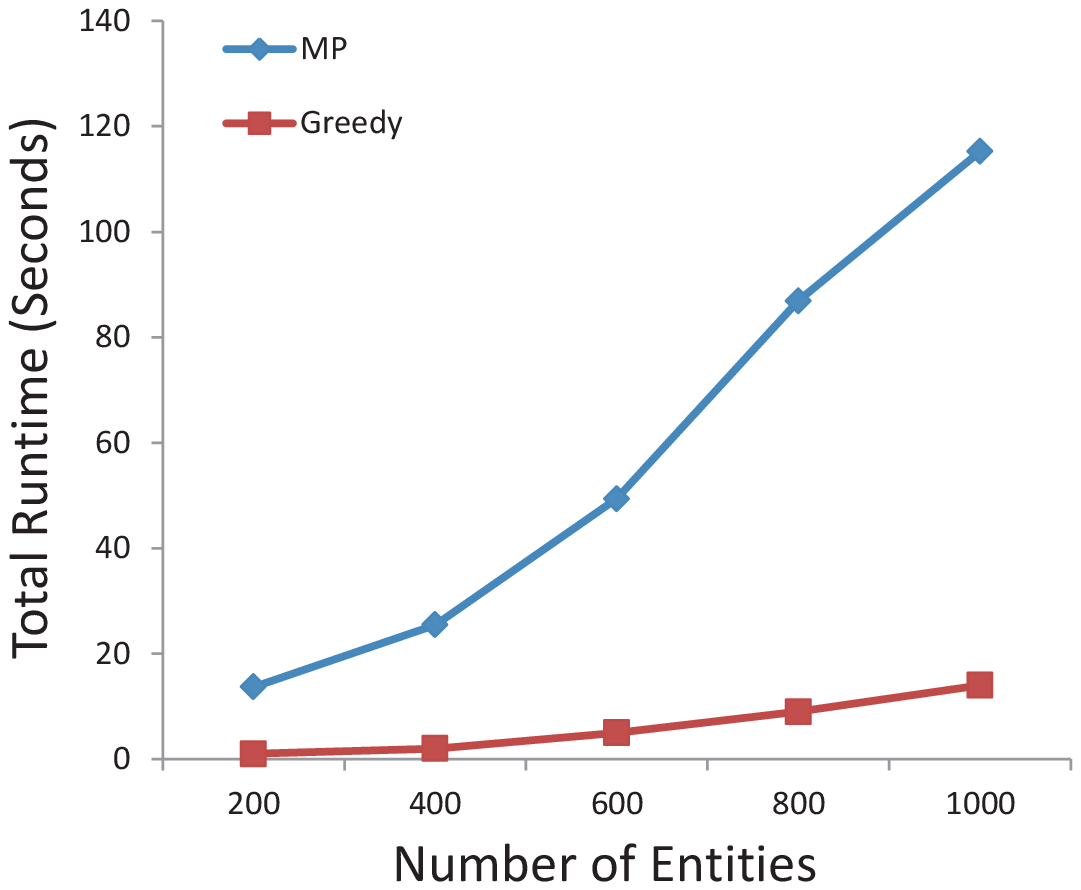} &
  \includegraphics[width=4cm]{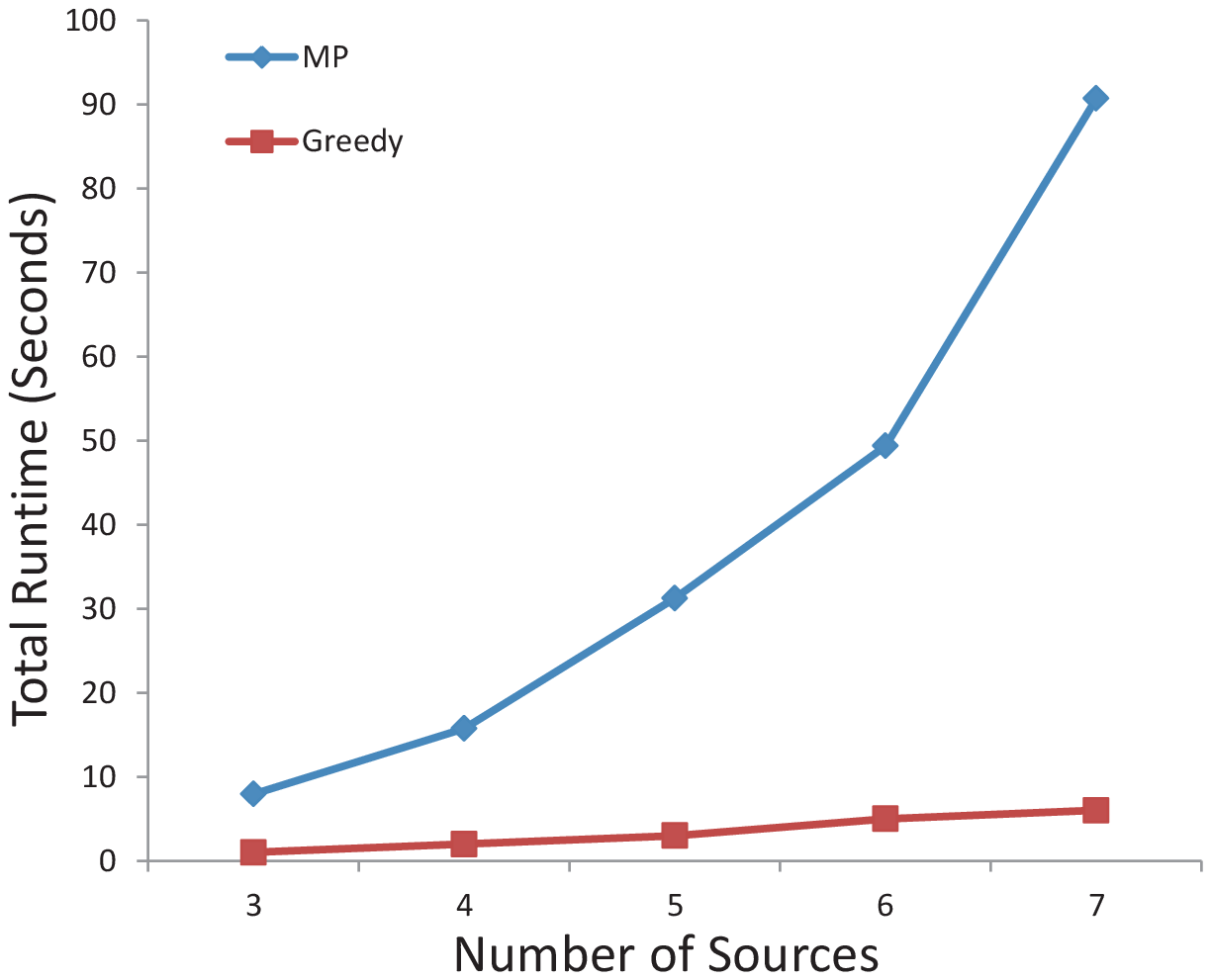} \\
  (a) Increase Entities & (b) Increase Sources
\end{tabular}
\caption{Running Time Comparison}
\label{fig:running_comparison}
\end{figure}

\section{Conclusions}\label{sec:conc}
In this paper, we have studied the multi-partite matching problem for integration across multiple data sources. We have proposed a sophisticated factor-graph message-passing algorithm and a greedy approach for solving the problem in the presence of one-to-one constraints, motivated by real-world socio-economic properties that drive data sources to be naturally deduplicated. We provided a competitive ratio analysis of the latter approach, and conducted comparisons of the message-passing and greedy approaches on a very large real-world Bing movie dataset, a smaller publications dataset, and synthetic data. Our experimental results prove that with additional sources, the precision and recall of entity resolution improve; that leveraging the global constraint improves resolution; and that message-passing, while slower to run, is much more robust to noisy data than the greedy approach.

For future work, implementing a parallelized version of the message-passing approach is an interesting direction to follow. Another open area is the formation of theoretical connections between the surrogate objective of weight maximization and the end-goal of high precision and recall.

% is not selected <OR> if conference mode is selected - because compsoc
% conference papers position the abstract like regular (non-compsoc)
\IEEEdisplaynotcompsoctitleabstractindextext

% \fi
\IEEEpeerreviewmaketitle

\bibliographystyle{abbrv}
\bibliography{sources}

\begin{IEEEbiography}[{\includegraphics[width=1in,height=1.25in,clip,keepaspectratio]{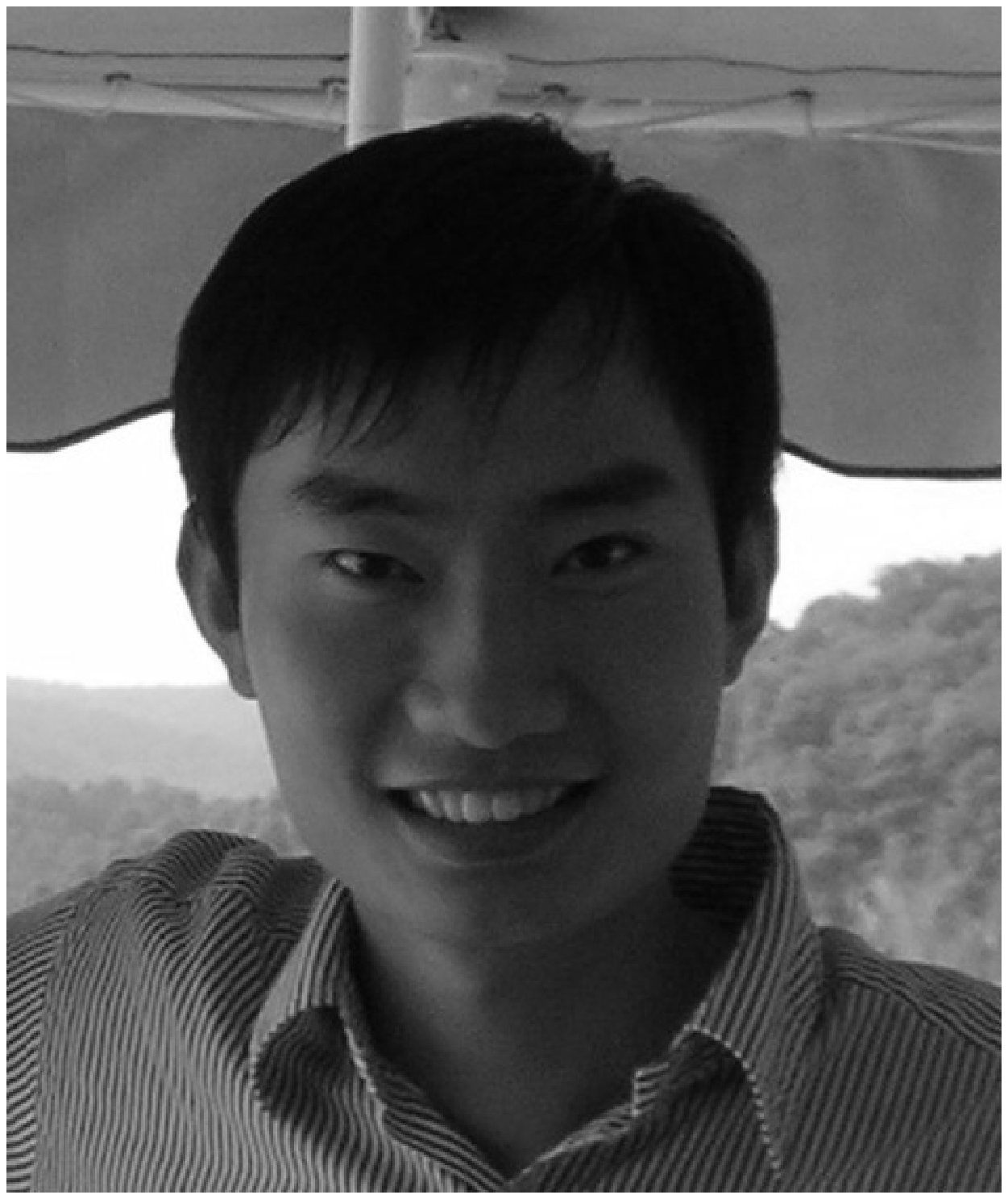}}]{Duo Zhang}
is a software engineer on the Ads team at Twitter. Before joining Twitter, he received his Ph.D. from the University of Illinois at Urbana-Champaign. He was a research intern at Microsoft, IBM, and Facebook during his Ph.D program. Dr. Zhang has published numerous research papers in text mining, information retrieval, databases, and social networking. He has also served on PCs and reviewers at major computer science conferences and journals including SIGKDD, ACL, and TIST.
\end{IEEEbiography}

\begin{IEEEbiography}[{\includegraphics[width=1in,height=1.25in,clip,keepaspectratio]{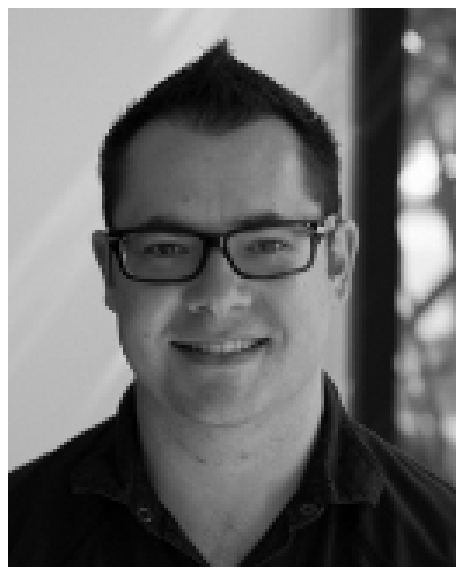}}]{Benjamin I. P. Rubinstein}
is Senior Lecturer in CIS at the University of Melbourne, Australia, and holds a PhD from UC Berkeley. He actively researches in statistical machine learning, databases, security \& privacy. Rubinstein has served on PCs and organised workshops at major conferences in these areas including ICML, SIGMOD, CCS. Previously he has worked in the research divisions of Microsoft, Google, Yahoo!, Intel (all in the US), and at IBM Research Australia. Most notably as a Researcher at MSR Silicon Valley Ben helped ship production systems for entity resolution in Bing and the Xbox360.
\end{IEEEbiography}

%\newpage

\begin{IEEEbiography}[{\includegraphics[width=1in,height=1.25in,clip,keepaspectratio]{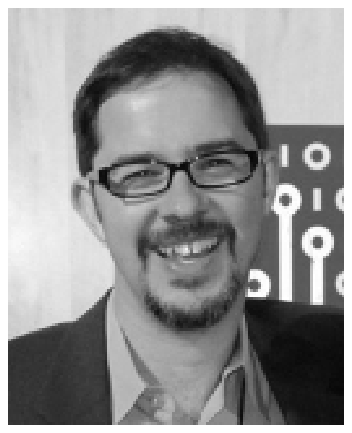}}]{Jim Gemmell}
is CTO of startup Tr\={o}v and holds PhD and M.Math degrees.
Dr. Gemmell is a world leader in the field of life-logging, and is author of the popular book \emph{Your Life, Uploaded}.
He has numerous publications in a wide range of areas including life-logging, multimedia, networking, video-conferencing, and databases.
Dr. Gemmell was previously Senior Researcher at Microsoft Research where he made leading contributions to major products
including Bing, Xbox360, and MyLifeBits.
\end{IEEEbiography}


%\vfill


\vfill

\end{document}